\PassOptionsToPackage{nosumlimits,nonamelimits}{amsmath}
\PassOptionsToPackage{colorlinks,linkcolor={blue},citecolor={blue},urlcolor={red},breaklinks=true,final}{hyperref}
\documentclass[a4paper,UKenglish,cleveref, Cref, thm-restate,numberwithinsect,final]{lipics-v2021}

\hideLIPIcs

\nolinenumbers

\numberwithin{equation}{section}

\usepackage{stel-common}

\usepackage{bm} %

\newcommand{\mybar}[3]{%
  \mathrlap{\hspace{#2}\overline{\scalebox{#1}[1]{\phantom{\ensuremath{#3}}}}}\ensuremath{#3}
}

\newcommand{\barB}{\mybar{0.6}{1.65pt}{B}}

\newcommand{\barSigmas}{\mybar{0.9}{0pt}{\Sigmas}}
\newcommand{\barrho}{\mybar{0.9}{1pt}{\rho}}

\newcommand{\barF}{\mybar{0.6}{2pt}{F}}

\newcommand{\barSigma}{\mybar{0.9}{1pt}{\Sigma}}

\newcommand{\tr}{\mathsf{tr}}

\newcommand{\Kl}{\mathbf{Kl}}

\newcommand{\DCPOb}{\mathbf{DCPO}_\bot}

\providecommand{\catname}{\mathbf} 
\providecommand{\clsname}{\mathcal}
\providecommand{\oname}[1]{{\mathop{\mathsf{#1}}\xspace}}

\def\defcatname#1{\expandafter\def\csname B#1\endcsname{\catname{#1}}}
\def\defcatnames#1{\ifx#1\defcatnames\else\defcatname#1\expandafter\defcatnames\fi}
\defcatnames ABCDEFGHIJKLMNOPQRSTUVWXYZ\defcatnames

\def\defclsname#1{\expandafter\def\csname C#1\endcsname{\clsname{#1}}}
\def\defclsnames#1{\ifx#1\defclsnames\else\defclsname#1\expandafter\defclsnames\fi}
\defclsnames ABCDEFGHIJKLMNOPQRSTUVWXYZ\defclsnames

\def\defbbname#1{\expandafter\def\csname BB#1\endcsname{{\mathbb{#1}}}}
\def\defbbnames#1{\ifx#1\defbbnames\else\defbbname#1\expandafter\defbbnames\fi}
\defbbnames ABCDEFGHIJKLMNOPQRSTUVWXYZ\defbbnames

\def\Set{\catname{Set}}

\DeclareOldFontCommand{\bf}{\normalfont\bfseries}{\mathbf}

\providecommand{\id}{\mathsf{id}}

\providecommand{\comp}{\mathbin{\circ}}

\providecommand{\bang}{\operatorname!}				             %

\providecommand{\xto}[1]{\,\xrightarrow{#1}\,}

\providecommand{\dar}{\kern-1.2pt\operatorname{\downarrow}}	
\providecommand{\uar}{\kern-1.2pt\operatorname{\uparrow}}

\providecommand{\fst}{\oname{fst}}

\providecommand{\brks}[1]{\langle #1\rangle}

\providecommand{\inl}{\oname{inl}}
\providecommand{\inr}{\oname{inr}}

\DeclareSymbolFont{Symbols}{OMS}{cmsy}{m}{n}
\DeclareMathSymbol{\iobj}{\mathord}{Symbols}{"3B}

\usepackage{stmaryrd}

\providecommand{\pacman}[1]{}					                     %

\newcommand{\undefine}[1]{\let #1\relax}					                       %

\providecommand{\mone}{{\text{\kern.5pt\rmfamily-}\mathsf{\kern-.5pt1}}}

\makeatletter
\def\mfix#1{\oname{#1}\@ifnextchar\bgroup\@mfix{}}	       %
\def\@mfix#1{#1\@ifnextchar\bgroup\mfix{}}			           %
\makeatother

\providecommand{\case}[3]{\mfix{case}{\mathbin{}#1}{of}{#2}{\kern-1pt;}{\mathbin{}#3}}

\DeclareMathSymbol{\mathinvertedexclamationmark}{\mathord}{operators}{'074}
\DeclareMathSymbol{\mathexclamationmark}{\mathord}{operators}{'041}
\makeatletter
\newcommand{\raisedmathinvertedexclamationmark}{%
  \mathord{\mathpalette\raised@mathinvertedexclamationmark\relax}%
}
\newcommand{\raised@mathinvertedexclamationmark}[2]{%
  \raisebox{\depth}{$\m@th#1\mathinvertedexclamationmark$}%
}
\makeatother

\newcommand{\R}{\mathcal{R}}

\newcommand{\Pow}{\mathcal{P}}

\newcommand{\hatini}{\hat{\ini}}

\newcommand{\qqand}{\qquad\text{and}\qquad}

\newcommand{\Sigmas}{\Sigma^{\star}}

\newcommand{\ar}{\mathsf{ar}}

\renewcommand{\S}{{\mathcal{S}}}

\newcommand{\seq}{\subseteq}
\newcommand{\ol}{\overline}

\newcommand{\outl}{\mathsf{outl}}
\newcommand{\outr}{\mathsf{outr}}

\newcommand{\beh}{{\mathsf{beh}}}
\providecommand{\B}{}
\providecommand{\C}{}
\providecommand{\D}{}

\renewcommand{\B}{{\mathcal{B}}}
\renewcommand{\C}{{\mathcal{C}}}
\renewcommand{\D}{{\mathcal{D}}}

\renewcommand{\id}{{\mathsf{id}}}
\newcommand{\Nat}{\mathds{N}}

\newcommand{\f}{\oname{f}}

\newcommand{\takeout}[1]{\empty}

\newcommand{\ini}{\iota}

\renewcommand{\rho}{\varrho}

\makeatletter
\newsavebox{\@brx}
\newcommand{\llangle}[1][]{\savebox{\@brx}{\(\m@th{#1\langle}\)}%
  \mathopen{\copy\@brx\kern-0.5\wd\@brx\usebox{\@brx}}}
\newcommand{\rrangle}[1][]{\savebox{\@brx}{\(\m@th{#1\rangle}\)}%
  \mathclose{\copy\@brx\kern-0.5\wd\@brx\usebox{\@brx}}}
\makeatother

\renewcommand{\comp}{\cdot}
\renewcommand{\c}{\colon}

\renewcommand{\Nat}{\mathbb{N}}

\newcommand{\mS}{{\mu\Sigma}}

\renewcommand{\epsilon}{\varepsilon}

\newcommand*\xbar[1]{%
  \kern.2em\hbox{%
    \vbox{%
      \hrule height 0.5pt %
      \kern0.5ex%
      \hbox{%
        \kern-0.2em%
        \ensuremath{#1}%
        \kern-0.4em%
      }%
    }%
  }\kern.4em %
} 

\makeatletter
\newcommand{\monto}{\@ifstar{\@mtolifted}{\@mto}}
\newcommand{\@mto}{\multimapdot}
\newcommand{\@mtolifted}{\mathbin{\xbar{\multimapdot}}}
\makeatother

\newcommand{\app}[2]{\,}

\makeatletter
	\newcommand{\pushright}[1]{\ifmeasuring@#1\else\omit\hfill$\displaystyle#1$\fi\ignorespaces}
	\newcommand{\pushleft}[1]{\ifmeasuring@#1\else\omit$\displaystyle#1$\hfill\fi\ignorespaces}
\makeatother

\usepackage{enumitem}
\setlist[enumerate,1]{label=(\arabic*),font=\normalfont,align=left,leftmargin=0pt,labelindent=0pt,listparindent=\parindent,labelwidth=0pt,itemindent=!,topsep=2pt,parsep=0pt,itemsep=2pt,start=1}
\setlist[enumerate,2]{label=(\alph*),font=\normalfont,labelindent=*,leftmargin=*,start=1}
\setlist[itemize]{labelindent=*,leftmargin=*}
\setlist[description]{labelindent=*,leftmargin=*,itemindent=-1 em}

\usepackage{microtype}

\usepackage{ifdraft}
\ifdraft{
  \usepackage[draft]{showlabels}
  
  \usepackage[footnote,marginclue,nomargin]{fixme}
}{
 \usepackage[final]{fixme}
}

\FXRegisterAuthor{hu}{ahu}{HU} %
\FXRegisterAuthor{sm}{asm}{SM} %
\FXRegisterAuthor{st}{ast}{ST} %
\FXRegisterAuthor{ls}{als}{LS} %
\FXRegisterAuthor{sg}{asg}{SG} %

\renewcommand{\comp}{\mathbin{\circ}}
\renewcommand{\c}{\colon}

\newcommand{\term}{1} %

\usepackage{seqsplit}
\usepackage{xstring}
\usepackage{xcolor}

\usepackage{centernot}
\usepackage{mathtools}

\tikzstyle{shiftarr}=[
        rounded corners,%
        to path={--([#1]\tikztostart.center)
                     -- ([#1]\tikztotarget.center) \tikztonodes
                     -- (\tikztotarget)},
]

\tikzset{
    commutative diagrams/.cd,
    arrow style=tikz,
    row sep=large,
    column sep = {10em}
}

\tikzset{cong/.style={draw=none,edge node={node [sloped, allow upside down, auto=false]{$\cong$}}},
         iso/.style={draw=none,every to/.append style={edge node={node [sloped, allow upside down, auto=false]{$\cong$}}}}}

\usetikzlibrary{decorations.pathmorphing}

\tikzcdset{scale cd/.style={
	every label/.append style={scale=#1},
    cells={nodes={scale=#1}}
}}
\usepackage{hypcap}
\def\resettheorembrackets{
\def\theorembracketopen{(}
\def\theorembracketclose{)}
}
\resettheorembrackets
\makeatletter
\def\@spopargbegintheorem#1#2#3#4#5{\trivlist
      \item[\hskip\labelsep{#4#1\ #2}]{#4{\theorembracketopen}#3{\theorembracketclose}\@thmcounterend\ }#5}
\makeatother

\newcommand{\resetCurThmBraces}{%
  \gdef\curThmBraceOpen{(}%
  \gdef\curThmBraceClose{)}}
\resetCurThmBraces
\newcommand{\removeThmBraces}{%
  \gdef\curThmBraceOpen{}%
  \gdef\curThmBraceClose{}}
\resetCurThmBraces

\newenvironment{notheorembrackets}{\removeThmBraces}{\resetCurThmBraces}

\usepackage{stackengine}
\stackMath
\newcommand\tsup[2][2]{%
 \def\useanchorwidth{T}%
  \ifnum#1>1%
    \stackon[-1.3ex]{\tsup[\numexpr#1-1\relax]{#2}}{\scalebox{2}[1]{$\mathchar"307E$}\kern-.5pt}%
  \else%
    \stackon[-1ex]{#2}{\scalebox{2}[1]{$\mathchar"307E$}\kern-.5pt}%
  \fi%
}

\newsavebox{\kleisliarrow}
\savebox{\kleisliarrow}{%
\begin{tikzpicture}[
      baseline=(arrow.base),
      inner sep=0mm,
      outer sep=0mm,
      ]
      \node[draw=none,
      anchor=base,
      overlay,
      inner sep=0,
      outer sep=0,
      minimum height=1em,
      ] (arrow) {$\phantom{\longrightarrow}$};

      \coordinate (circle pos) at
        ($ (arrow.south) !.58! (arrow.north)$);
      \begin{scope}[even odd rule,overlay]
        \clip  (circle pos) circle (0.17em)
           (arrow.north west) rectangle (arrow.south east);
      \node[draw=none,
      anchor=base,
      inner sep=0,
      outer sep=0,
      ] (arrow) {$\longrightarrow$};
    \end{scope}
    \draw[fill=black,overlay]
      ($ (arrow.south) !.58! (arrow.north)$)
      circle (0.15em);
    \draw[use as bounding box,draw=none] (arrow.north west) rectangle (arrow.south east);
  \end{tikzpicture}}

\newcommand{\kleislito}{\ensuremath{\mathbin{\usebox{\kleisliarrow}}}}

\newsavebox{\kleislidot}
\savebox{\kleislidot}{%
\begin{tikzpicture}[baseline=0pt,outer sep=0pt]
    \draw[fill=black,solid] (0,0) circle (1.5pt);
  \end{tikzpicture}}

\tikzstyle{kleisli}=[
"{\usebox{\kleislidot}}"{red,anchor=center,font=\normalsize,pos=0.5},
outer sep = 1pt, %
]

\newsavebox{\kleisliarroww}
\savebox{\kleisliarroww}{%
\begin{tikzpicture}[
      baseline=(arrow.base),
      inner sep=0mm,
      outer sep=0mm,
      ]
      \node[draw=none,
      anchor=base,
      overlay,
      inner sep=0,
      outer sep=0,
      minimum height=1em,
      ] (arrow) {$\phantom{\longrightarrow}$};

      \coordinate (circle pos) at
        ($ (arrow.south) !.58! (arrow.north)$);
      \begin{scope}[even odd rule,overlay]
        \clip  (circle pos) circle (0.17em)
           (arrow.north west) rectangle (arrow.south east);
      \node[draw=none,
      anchor=base,
      inner sep=0,
      outer sep=0,
      ] (arrow) {$\longrightarrow$};
    \end{scope}
    \draw[fill=none,overlay]
      ($ (arrow.south) !.58! (arrow.north)$)
      circle (0.15em);
    \draw[use as bounding box,draw=none] (arrow.north west) rectangle (arrow.south east);
  \end{tikzpicture}}

\newsavebox{\kleisliarrowb}
\savebox{\kleisliarrowb}{%
\begin{tikzpicture}[
      baseline=(arrow.base),
      inner sep=0mm,
      outer sep=0mm,
      ]
      \node[draw=none,
      anchor=base,
      overlay,
      inner sep=0,
      outer sep=0,
      minimum height=1em,
      ] (arrow) {$\phantom{\longrightarrow}$};

      \coordinate (circle pos) at
        ($ (arrow.south) !.58! (arrow.north)$);
      \begin{scope}[even odd rule,overlay]
        \clip  (circle pos) circle (0.17em)
           (arrow.north west) rectangle (arrow.south east);
      \node[draw=none,
      anchor=base,
      inner sep=0,
      outer sep=0,
      ] (arrow) {$\longrightarrow$};
    \end{scope}
    \draw[fill=none,overlay]
      ($ (arrow.south) !.58! (arrow.north)$)
      circle (0.15em);
    \draw[use as bounding box,draw=none] (arrow.north west) rectangle (arrow.south east);
  \end{tikzpicture}}

\newcommand{\kleislitow}{\ensuremath{\mathbin{\usebox{\kleisliarroww}}}}

\newsavebox{\kleislidotw}
\savebox{\kleislidotw}{%
\begin{tikzpicture}[baseline=0pt,outer sep=0pt]
    \draw[fill=white,solid] (0,0) circle (1.5pt);
  \end{tikzpicture}}

\tikzstyle{kleisliw}=[
"{\usebox{\kleislidotw}}"{red,anchor=center,font=\normalsize,pos=0.5},
outer sep = 1pt, %
]

\newsavebox{\kleislidotb}
\savebox{\kleislidotb}{%
\begin{tikzpicture}[baseline=0pt,outer sep=0pt]
    \draw[fill=white,solid] (0,0) circle (1.5pt);
  \end{tikzpicture}}

\tikzstyle{kleislib}=[
"{\usebox{\kleislidotb}}"{red,anchor=center,font=\normalsize,pos=0.5},
outer sep = 1pt, %
]

\renewcommand{\kleislito}{\kleislitow}

\ifdraft{
  \newcommand{\ST}[1]{\textcolor{purple}{ST: #1}}
  \newcommand{\SG}[1]{\textcolor{red!50!magenta}{\bfseries SG: #1}}
  \newcommand{\RJ}[1]{\textcolor{green!75!black}{\bfseries RJ: #1}}
  \newcommand{\JF}[1]{\textcolor{orange}{JF: #1}}
  \newcommand{\STin}[1]{\todo[color=purple!30,inline]{Stelios: #1}}
  \newcommand{\SGin}[1]{\todo[color=orange!60,inline]{Sergey: #1}}
}{
  \newcommand{\ST}[1]{}
  \newcommand{\SG}[1]{}
  \newcommand{\RJ}[1]{}
  \newcommand{\JF}[1]{}
  \newcommand{\STin}[1]{}
  \newcommand{\SGin}[1]{}
}

\usepackage{xspace}

\usepackage{etoolbox} %

\theoremstyle{plain}

\theoremstyle{definition}
\newtheorem{defn}[theorem]{Definition} %
\newtheorem{rem}[theorem]{Remark} %

\newtheorem{assumptions}[theorem]{Assumptions}
\newtheorem{construction}[theorem]{Construction}
\crefname{expl}{Example}{Examples}
\crefname{defn}{Definition}{Definitions}
\crefname{construction}{Construction}{Constructions}
\crefname{assumptions}{Assumptions}{Assumptions}

\newcommand{\W}{{\mathcal{W}}}

\let\oldcheckmark\checkmark
\renewcommand{\checkmark}{\raisebox{-4pt}{\scalebox{1.2}[.65]{$\oldcheckmark$}}}

\usepackage{tensor}

\makeatletter
\DeclareRobustCommand\widecheck[1]{{\mathpalette\@widecheck{#1}}}
\def\@widecheck#1#2{%
  \setbox\z@\hbox{\m@th$#1#2$}%
  \setbox\tw@\hbox{\m@th$#1%
    \widehat{%
      \vrule\@width\z@\@height\ht\z@
      \vrule\@height\z@\@width\wd\z@}$}%
  \dp\tw@-\ht\z@ppp
  \@tempdima\ht\z@ \advance\@tempdima2\ht\tw@ \divide\@tempdima\thr@@
  \setbox\tw@\hbox{%
    \raise\@tempdima\hbox{\scalebox{1}[-1]{\lower\@tempdima\box
        \tw@}}}%
  {\ooalign{\box\tw@ \cr \box\z@}}}
\makeatother

\makeatletter
\newcommand{\superimpose}[2]{{%
  \ooalign{%
    \hfil$\m@th#1\@firstoftwo#2$\hfil\cr
    \hfil$\m@th#1\@secondoftwo#2$\hfil\cr
  }%
}}
\makeatother

\newcommand{\while}{\ensuremath{\mathsf{Imp}}\xspace}
\newcommand{\whiletwo}{\ensuremath{\mathsf{Imp}^{\mathsf{2}}}\xspace}

\newcommand{\store}{{\mathcal{S}}}

\newcommand{\bZ}{\mathbb{Z}}
\newcommand{\expr}{\mathsf{Ex}}
\newcommand{\Progs}{\mathsf{P}}

\newcommand{\asn}{\mathrel{\coloneqq}}

\newcommand{\seqcomp}{\mathop{;}}

\newcommand{\Tr}{\mathsf{R}}
\newcommand{\Co}{\mathsf{W}}

\newcommand{\Trs}{{\textsf{r}}}
\newcommand{\Cos}{{\textsf{w}}}

\newcommand{\run}[2]{[#2]_{#1}}

\theoremstyle{definition}
\newtheorem{notation}[theorem]{Notation}

\usepackage{scalerel,stackengine,graphicx}

\usepackage[normalem]{ulem}



\def\star{\ast} %

\author{Robin Jourde}{Université Savoie Mont Blanc, France}{robin.jourde@univ-smb.fr}{https://orcid.org/0009-0003-5787-8463}{}

\author{Henning Urbat}{Friedrich-Alexander-Universität
Erlangen-Nürnberg, Germany}{henning.urbat@fau.de}{https://orcid.org/0000-0002-3265-7168}{}

\author{Sergey Goncharov}{University of Birmingham, UK}{s.goncharov@bham.ac.uk}{https://orcid.org/0000-0001-6924-8766}{Funded by the Deutsche Forschungsgemeinschaft (DFG, German
  Research Foundation) -- project number 527481841 (ATLaS)}
\author{Stelios Tsampas}{
Syddansk Universitet, Denmark}{stelios@imada.sdu.dk}{https://orcid.org/0000-0001-8981-2328}{Funded by the Deutsche Forschungsgemeinschaft (DFG, German Research
Foundation) -- project number 527481841 (ATLaS)}

\author{Jonas Forster}{Friedrich-Alexander-Universität
Erlangen-Nürnberg, Germany}{jonas.forster@fau.de}{https://orcid.org/0000-0002-5050-2565}{Funded by the Deutsche Forschungsgemeinschaft (DFG, German Research
Foundation) -- project number 434050016 (SpeQt)}

\authorrunning{Jourde et al.}

\title{Compositionality in Coalgebraic Trace Semantics}

\begin{document}
\allowdisplaybreaks

\sloppy

\maketitle
\begin{abstract}
A key requirement on any well-behaved process language is its \emph{compositionality}: behavioural equivalence of processes should be respected by the constructors of the language. Turi and Plotkin's \emph{abstract GSOS} provides an elegant bialgebraic framework for modelling rule formats that guarantee compositionality from the outset. Their original results, however, are restricted to compositionality of strong bisimilarity, a rather fine-grained notion of process equivalence. In the present paper, we demonstrate that Turi and Plotkin's approach also applies to \emph{trace equivalence}, which only observes external actions of processes. To this end, we revisit the general compositionality result of their original theory and present it in a refined form with regard to the required naturality conditions. This step makes abstract GSOS applicable over {Kleisli categories} and thereby enables reasoning about compositionality in the setting of \emph{coalgebraic trace semantics}. As our main contribution, we introduce \emph{De Simone laws}, a type of GSOS laws over Kleisli categories, and prove that their operational models are compositional for coalgebraic trace equivalence. This result recovers and explains compositionality of the well-known \emph{De Simone} rule format for labelled transition systems in a natural categorical setting. As a further application, we derive from our general framework a novel De Simone-type format for \emph{probabilistic} systems, compositional for probabilistic trace equivalence.
\end{abstract}


\section{Introduction}\label{sec:intro}
One of the primary challenges in concurrency theory is the development of effective tools and techniques to reason about \emph{equivalence} of processes. Numerous notions of process equivalence at different levels of granularity have been studied. Among the most prominent ones, lying at opposite ends of the linear-time branching-time spectrum~\cite{glabbeek90}, are \emph{strong bisimilarity}, corresponding to a low-level observer with full access to internal states and transitions of the given processes, and \emph{trace equivalence}, where only external actions are observed. 
 Regardless of the chosen flavour of equivalence, a key desideratum is that it forms a \emph{congruence}, which means that it is respected by all process constructors. For instance, one would expect two processes $p\parallel q$ and $p'\parallel q$ to be equivalent if their subprocesses $p$ and $p'$ are. The importance of this property is that it enables \emph{compositional} reasoning: it ensures that the observable behaviour of a complex process is fully determined by that of each individual component.

However, actually \emph{proving} that a process language is compositional for a given type of equivalence is cumbersome, in particular in advanced settings with higher-order features, computational effects, etc. Hence, there has been longstanding interest in syntactic rule formats that guarantee compositionality for all languages whose specification adheres to the given format~\cite{DBLP:books/el/01/AcetoFV01}. Well-known examples include the \emph{GSOS}~\cite{DBLP:journals/jacm/BloomIM95} and \emph{tyft/tyxt} formats~\cite{DBLP:journals/iandc/GrooteV92}, which are compositional for strong bisimilarity, and the \emph{De Simone} format~\cite{DeSimone1985,bloom94}, a fragment of GSOS compositional for trace semantics. Generally, coarse-grained process equivalences with very limited internal access to processes, such as trace equivalence, tend to require severe syntactic restrictions on the operational rules to ensure compositionality. Standardly, those restrictions are discovered ad hoc, as are the proofs of their congruence properties.

In the present paper, we develop a systematic \emph{categorical} approach to congruence formats for trace equivalence that treats the given process language and the structural rules specifying its operational semantics as abstract parameters. Our work draws from two different theories:

\medskip\noindent\textbf{Coalgebraic Trace Semantics.} Labelled transition systems and other types of state-based structures that underlie the operational semantics of process languages admit a uniform categorical abstraction in the form of \emph{coalgebras}~\cite{DBLP:journals/tcs/Rutten00} for a functor encapsulating the system type. Hasuo et al.~\cite{hjs07} have shown that various notions of trace semantics can be captured at coalgebraic generality by considering final coalgebras over order-enriched \emph{Kleisli categories}.

\medskip\noindent\textbf{Abstract GSOS.} Turi and Plotkin~\cite{DBLP:conf/lics/TuriP97} have proposed an elegant categorical approach to GSOS-type rule formats. The idea is to present GSOS rules as \emph{GSOS laws}, a type of natural transformation that distributes the \emph{syntax} of a language over its \emph{behaviour}, both modelled by functors. The core selling point  of their framework is its free compositionality theorem: for every language modelled by a GSOS law, behavioural equivalence (i.e.\ equality in the final coalgebra of the behaviour functor) is a congruence. The compositionality of GSOS rules w.r.t.~strong bisimilarity~\cite{DBLP:journals/jacm/BloomIM95} is one of many instances~\cite{DBLP:conf/ctcs/Turi97,DBLP:journals/tcs/Klin11} of this general result.

\medskip\noindent In the light of the above work, we base our abstract theory of congruence for traces on a (seemingly) simple strategy: run the machinery of abstract GSOS over Kleisli categories, the setting for coalgebraic trace semantics. As a first proof of concept, we demonstrate that just like GSOS rules are captured by GSOS laws over the category of sets, De Simone rules correspond to \emph{De Simone laws}, a novel type of GSOS laws living in the Kleisli category of an additive monad. Our key insight is that the ad hoc syntactic restrictions imposed on De Simone rules (regarding affineness of terms) are precisely reflected by certain naturality conditions on their associated De Simone laws. At this point, a technical challenge arises: Since De Simone laws are only natural w.r.t.~\emph{some} Kleisli morphisms, Turi and Plotkin's theory of congruence (which is tailored to fully natural GSOS laws) does not immediately apply to this setting. Therefore, we revisit Turi and Plotkin's original congruence theorem for abstract GSOS and present it in a refined form that analyses the precise required naturality conditions (\Cref{thm:cong}). Since the original proof~\cite{DBLP:conf/lics/TuriP97} heavily relies on naturality, the refined congruence result requires a conceptually different approach beyond mere proof inspection.

The refined congruence theorem for abstract GSOS allows us to establish our main result (\Cref{thm:cong-trace-positive}): coalgebras specified by De Simone laws admit compositional trace semantics. Besides incorporating the congruence property for classical De Simone rules~\cite{bloom94} into the abstract GSOS framework, this result gives rise to novel congruence formats for \emph{effectful} settings, thus highlighting the scope and effectiveness of our theory. As an advanced case study, we derive a De Simone-type format for (almost-surely terminating) \emph{probabilistic transition systems}. While several congruence formats for probabilistic or weighted bisimilarity have emerged in the literature~\cite{bartels02,ks13,CASTIGLIONI2024100929}, our probabilistic De Simone format is (to the best of our knowledge) the first congruence format for probabilistic trace semantics. 

\subparagraph*{Related Work}
Since the introduction of abstract GSOS~\cite{DBLP:conf/lics/TuriP97}, a categorical understanding of congruence formats for equivalence notions coarser than strong bisimilarity has been pursued. Most previous approaches in this direction conceptually depart from Turi and Plotkin's work (rather than using it) and develop the respective theory of congruence largely from scratch.

 Klin~\cite{Klin09,Klin10} presents an abstract approach to compositional semantics parametric in a \emph{coalgebraic logic}, a form of dual adjunction that supplies the targeted notion of semantics. Technically, this work is not compatible with ours, although some instances (such as compositionality of standard De Simone rules) are shared. Klin's logical framework is highly flexible and expressive. One drawback is that its application requires, on top of the coalgebraic logic, coming up with a further technical ingredient in the form of a \emph{logical distributive law}, which is typically not trivial to find and subject to a complex compatibility condition. This contrasts our approach where only the standard data of coalgebraic trace semantics is needed to instantiate the theory, and the technical conditions of the congruence result amount to naturality conditions over a Kleisli category with a simple underlying intuition.

The interaction of abstract GSOS with monadic, or generally effectful,
behaviour has been a topic of
interest for a long time~\cite{DBLP:conf/fossacs/Abou-SalehP13,56f40c248cb44359beb3c28c3263838e,gmstu25,DBLP:conf/fscd/0001MS0U22,DBLP:journals/tcs/MiculanP16,10.1007/978-3-642-32784-1_5,10.1145/3776697},
the common goal typically being compositionality for various effectful
program equivalences. Notably, Abou-Saleh and Pattinson develop a form of abstract GSOS over a \emph{mixed} Kleisli
setting~\cite{DBLP:journals/entcs/Abou-SalehP11,DBLP:phd/ethos/AbouSaleh14},
and study the semantics given by final coalgebras in a Kleisli category. They propose
a (rather technical and hard to verify)~\emph{condition on cones} to guarantee compositionality. Unlike our own Kleisli-based approach, their theory of congruence is not based on Turi and Plotkin's results. Moreover, in terms of applications their  work is primarily aimed at reasoning about stateful imperative programming languages and does not encompass De Simone-type specifications.

Finally, Tsampas et
al.\,\cite{DBLP:conf/mfcs/0001WNDP21} consider GSOS laws featuring behaviours of the form
$TB$, where $T$ is a monad whose Kleisli category is order-enriched in a
manner that $TB$-coalgebras can be
\emph{saturated}\,\cite{DBLP:journals/jlp/BrengosMP15}. Their theory of congruence applies to
weak bisimilarity rather than trace semantics and also departs from the standard
theory of Turi and Plotkin.

\section{Categorical Preliminaries}
Our theory of compositional trace semantics rests on various abstractions using the language of category theory~\cite{mac2013categories}, notably (co)algebras and Kleisli extensions reviewed below. Readers are assumed to be familiar with functors, natural transformations, (co)products, and monads.

\subparagraph*{Notation} Given
 objects $X_1, X_2$ in a category~$\C$, we write $X_1\times X_2$ for their product, $\outl\colon X_1\times X_2\to X_1$ and $\outr\colon X_1\times X_2\to X_2$ for the projections, and $\brks{f_1,f_2}\colon Y \to X_1 \times X_2$ for the pairing of
morphisms $f_i\colon Y \to X_i$, $i = 1,2$. We write $1$ for the terminal object and $!\colon X\to 1$ for the unique morphism. Dually, $X_1+X_2$ denotes the coproduct, $\inl\colon X_1\to X_1+X_2$ and $\inr\colon X_2\to X_1+ X_2$ the injections, and $[f_1,f_2]\colon X_1+X_2\to Y$ a co-pairing. 
In the category $\Set$ of sets and functions, products and coproducts are cartesian products and disjoint unions, resp., and $1=\{*\}$.
Given a monad $T$, we write $\eta^T$ and $\mu^T$ for the unit and multiplication; the superscript $T$ is usually dropped. 
A \emph{pre-natural transformation} between functors $F,G\colon \C\to \D$ is a family of morphisms
$\alpha_X\colon FX \to GX$ ($X\in \C$). 
It is \emph{natural w.r.t.\ a morphism $f\colon X\to Y$} of $\C$ if $\alpha_Y\comp Ff = Gf\comp \alpha_X$. 
A natural transformation is thus a pre-natural transformation that is natural w.r.t.\ all morphisms. 
We write $\alpha$ for $\alpha_X, \alpha_Y,\ldots$

\subparagraph*{Algebras} Algebraic structures admit a categorical abstraction in the form of algebras for an endofunctor $F$ on a category $\C$. An \emph{$F$-algebra} is a pair $(A,a)$ of an object~$A$
 and a morphism $a\colon F A\to A$. A \emph{morphism} from
$(A,a)$ to an $F$-algebra $(B,b)$ is a morphism $h\colon A\to B$
of~$\C$ with $h\comp a = b\comp F h$. Algebras for $F$ and
their morphisms form a category, and an \emph{initial}
$F$-algebra is simply an initial object in that category. We denote the initial algebra by $(\mu F,\ini)$ if it exists. By Lambek's Lemma~\cite{Lambek68}, its
structure $\ini\colon F(\mu F) \to \mu F$ is an isomorphism.

A \emph{free $F$-algebra} on an object $X$ of $\C$ is an
$F$-algebra $(F^{\star}X,\iota_X)$ together with a morphism
$\eta_X\c X\to F^{\star}X$ of~$\C$ such that for every algebra $(A,a)$
and every morphism $h\colon X\to A$ of $\C$, there is a unique
$F$-algebra morphism $h^\star\colon (F^{\star}X,\iota_X)\to (A,a)$
with $h=h^\star\comp \eta_X$. If $\C$ has an initial object $0$, then $F^{*} 0=\mu F$. If free algebras
exist on every object, their formation induces a monad
$F^{\star}\colon \C\to \C$, the \emph{free monad} generated by $F$~\cite{Barr70}. Every $F$-algebra $(A,a)$ yields an
Eilenberg-Moore algebra $\hat a:=\id_A^\star \colon F^{\star} A \to A$, where $\id_A\c A\to A$ is the identity.

\begin{example}[Algebras for a Signature] An (\emph{algebraic}) \emph{signature} consists of a set~$\Sigma$
of \emph{operation symbols} and a map $\ar\colon \Sigma\to \Nat$
associating to every $\f\in \Sigma$ its \emph{arity}. Every
signature~$\Sigma$ induces an endofunctor on the category $\Set$, denoted by
the same letter $\Sigma$ and defined by $\Sigma X = \coprod_{\f\in \Sigma} X^{\ar(\f)}$. Endofunctors of this form are called \emph{polynomial}. An algebra for the polynomial functor $\Sigma$ is
precisely an algebra for the signature $\Sigma$ in the usual sense, that is, a set $A$ equipped with an operation $\f^A\colon A^n\to A$ for every $n$-ary $\f\in\Sigma$. Morphisms of $\Sigma$-algebras are
maps respecting the algebraic structure. Given a set $X$, the free algebra $\Sigmas X$ 
is carried by the set of $\Sigma$-terms with variables from~$X$. The free
algebra on the empty set is the initial algebra~$\mu \Sigma$; it is
formed by all \emph{closed} $\Sigma$-terms. For every
$\Sigma$-algebra $(A,a)$, the Eilenberg-Moore algebra
$\hat{a}\colon \Sigmas A \to A$ evaluates terms in~$A$.

\end{example}

\subparagraph*{Coalgebras} The dual concept of algebras are coalgebras, which form a categorical abstraction of state-based transition systems. A \emph{coalgebra} for an
endofunctor $B$ on $\C$ is a pair $(C,c)$ of an object $C$ and a morphism $c\colon C\to BC$. A \emph{morphism} from
$(C,c)$ to a $B$-coalgebra $(D,d)$ is a morphism
$h\colon C\to D$ of $\C$ such that $Bh\comp c = d\comp h$.
Coalgebras for $B$ and their morphisms form a category, and a
\emph{final} $B$-coalgebra, denoted $(\nu B,\tau)$, is a final object in that category.

\begin{example}\label{ex:lts}
Labelled transition systems (LTS) with or without explicit termination correspond, respectively, to coalgebras of type
$c\colon C\to \Pow (L\times C+1)$ and $c_0\colon C\to \Pow(L\times C)$,
for a set $L$ of labels. Ordinary LTS (i.e.\ without explicit termination) form a subclass of LTS with explicit termination by regarding all states as terminating: put $c(x)=c_0(x)\cup \{*\}$.
\end{example}

\subparagraph*{Kleisli Extensions}
Given a monad $T$ on a category $\C$, we write $\Kl(T)$ for its Kleisli category. Objects of $\Kl(T)$ are those of $\C$, and a morphism $f\colon X\kleislitow Y$ of $\Kl(T)$ is a morphism $f\colon X\to TY$ of $\C$, with identities and composition induced by the monad structure. A \emph{Kleisli}\\[.2ex]
\begin{minipage}[t]{.69\textwidth}
\emph{extension} of an endofunctor $F\colon \C\to \C$ is an endofunctor $\barF\colon \Kl(T)\to \Kl(T)$ making the diagram on the right commute. Here $J_T$ is the canonical left adjoint, acting on objects as the identity and on morphisms $f\colon X \to Y$ by $J_Tf=\eta_Y\comp f$.

\end{minipage}
\begin{minipage}[t]{.3\textwidth}
\vspace{-.7cm}
\[  
\begin{tikzcd}
\Kl(T) \ar{r}{\barF} & \Kl(T) \\
\C \ar{u}{J_T} \ar{r}{F} & \C \ar{u}[swap]{J_T} 
\end{tikzcd}
\] 
\end{minipage}\\[.3ex]
Morphisms of $\Kl(T)$ of the form $J_T f$ are called \emph{pure}. We usually put $J=J_T$ if the monad is clear from the context. Kleisli extensions of an endofunctor $F$ are closely related to \emph{distributive laws} of $F$ over the monad $T$. The latter are natural transformations
$\delta\colon FT\to TF$    
 subject to compatibility conditions~\cite{mulry93}.
Each distributive law $\delta \colon FT \to TF$ induces a Kleisli extension $\barF\colon \Kl(T)\to \Kl(T)$, defined on objects as $\barF X = FX$ and on morphisms $f\colon X \kleislito Y$ as $\barF f = \delta_Y \comp Ff$.
Conversely, each Kleisli extension of $F$ arises from a unique distributive law. Thus there is a bijective correspondence between distributive laws and Kleisli extensions~\cite{mulry93}.

\begin{proposition}[\cite{hjs07}]\label{prop:initial-alg-lift}
Let $F\colon \C\to \C$ be an endofunctor with an initial algebra $\mS$ and a free algebra $F^{*} X$ on $X\in \C$, and let
$\bar F\colon \Kl(T)\to \Kl(T)$ be a Kleisli extension. Then 
$\mu \barF = \mu F$ and $\barF^{*}X = F^{*} X$,
with the $\barF$-algebra structures and universal maps given by
\[ \begin{tikzcd} \barF(\mu F) = F(\mu F) \ar[kleisliw]{r}{J\ini}  & \mu F, \end{tikzcd}\qquad \begin{tikzcd} \barF F^{*} X = FF^{*}X \ar[kleisliw]{r}{J\ini_X}  & F^{*}X, \end{tikzcd}\qquad \begin{tikzcd} X \ar[kleisliw]{r}{J\eta_X}  & F^{*}X.\end{tikzcd}\]
\end{proposition}

\section{Coalgebraic Trace Semantics}\label{sec:coalgebraic-trace-semantics}
We recall the Kleisli-based framework for coalgebraic trace semantics~\cite{hjs07}, following the streamlined presentation by Frank et al.~\cite{fmu22} which simplifies some of the technical conditions of the original paper. The framework applies to coalgebras of type
\begin{equation}\label{eq:tb-coalg} c\colon C\to TBC \end{equation}
where $B$ is an endofunctor (modelling input/output behaviour) and $T$ is a monad 
(modelling a computational effect such as non-determinism or probabilities) on a category $\C$.

\begin{example}[LTS]\label{ex:lts-trace} LTS with explicit termination are coalgebras of type \eqref{eq:tb-coalg} for
$T=\Pow$ (the power set monad) and $BX=L\times X+1$ on $\Set$.
We write
$x\xto{a} y$ if $(a,y)\in c(x)$, and $x{\downarrow}$ if $*\in c(x)$. A \emph{completed trace} at a state $x\in C$ is a string $a_1\cdots a_n\in L^{*}$  such that
$x=x_0\xto{a_1} x_1 \xto{a_2} \cdots \xto{a_{n-1}} x_{n-1}\xto{a_n}x_n{\downarrow}$ for some $x_0,\ldots,x_n\in C$. Let $\tr(x)\seq L^{*}$ be the set of completed traces at $x$. Note that $L^{*}$ carries the initial algebra $\mu B$ for the functor $B$. Thus completed trace semantics yields a morphism
$\tr\colon C\kleislitow \mu B$ in $\Kl(\Pow)$. 
For an ordinary LTS (with all states terminating), the map $\tr$ sends a state $x$ to the set of \emph{partial traces} at $x$. 
\end{example}

\noindent To generalize the trace map $\tr$ from LTS to coalgebras \eqref{eq:tb-coalg}, some order-enrichment on the Kleisli category $\Kl(T)$ is required. A category $\B$ is \emph{left strictly $\DCPOb$-enriched} if each hom-set $\B(X,Y)$ carries the structure of a DCPO with $\bot$ such that composition is left-strict ($\bot\comp f = \bot$) and preserves directed joins in each component ($f\comp \bigsqcup_i g_i = \bigsqcup_i f\comp g_i$ and $ (\bigsqcup_i g_i)\comp f= \bigsqcup_i g_i\comp f$).  
A functor $F\colon \B\to \B$ is \emph{locally monotone} if 
$f\sqsubseteq g$ implies $F f\sqsubseteq F g$.

\begin{defn}\label{def:trace-setting}
A \emph{setting for coalgebraic trace semantics} consists of a monad $T$ and an endofunctor $B$ on a category $\C$ with a Kleisli extension $\barB\colon \Kl(T)\to \Kl(T)$ such that:
\begin{enumerate}
\item the functor $B$ has an initial algebra $(\mu B, \beta)$;
\item\label{asm-trace-order1} the category $\Kl(T)$ is \emph{left strictly $\DCPOb$-enriched};
\item\label{asm-trace-order2} the functor $\barB\colon \Kl(T)\to \Kl(T)$ is \emph{locally monotone}.
\end{enumerate}
\end{defn}
The key property of any setting for coalgebraic trace semantics is the initial algebra/final coalgebra coincidence
$\mu B = \mu \barB = \nu \barB$. The first equality is \Cref{prop:initial-alg-lift}; the second one is the next proposition, which also appears  (with stronger conditions on $\C$ and $B$) in \cite{hjs07}.

\begin{proposition}[\cite{fmu22}]\label{prop:initial-alg-final-coalg}
Let $T$, $B$, $\barB$ be a setting for a coalgebraic trace semantics. Then the functor $\barB\colon \Kl(T)\to \Kl(T)$ has a final coalgebra given by 
$\begin{tikzcd} \mu B \ar[kleisliw]{r}{J\beta^{-1}}  & B(\mu B) = \barB(\mu B). \end{tikzcd}$
\end{proposition} 
(Recall that $\beta$ is an isomorphism by Lambek's lemma.) To obtain coalgebraic trace semantics, note that a $TB$-coalgebra \eqref{eq:tb-coalg} in $\C$ is the same as a $\barB$-coalgebra $c\colon C\kleislitow \barB C$ in $\Kl(T)$.
The \emph{trace map} for $c$ is the unique $\barB$-coalgebra morphism to the final coalgebra $(\mu B, J\beta^{-1})$:
\begin{equation}\label{eq:trace-map} \tr_c\colon C\kleislitow \mu B.
\end{equation} 
We usually drop the subscript if $c$ is clear from the context.

\begin{example}[LTS~\cite{hjs07}]\label{ex:lts-trace-coalg}
LTS correspond to the setting for coalgebraic trace semantics where $T=\Pow$ and $BX=L\times X+1$ (\Cref{ex:lts-trace}), and $\barB$ is induced by the distributive law
\[ \delta^{B}\colon L\times \Pow X + 1 \to \Pow(L\times X +1),\qquad 
(a,S) \mapsto\{ (a,s) \mid s\in S \},\qquad
* \mapsto \{*\}. \]
Then all the required conditions of \Cref{def:trace-setting}  are satisfied: 
\begin{enumerate}
\item The initial algebra for $B$ is given by the set $L^{*}$ with structure $\beta\colon L\times L^{*} + 1 \to L^{*}$, defined componentwise by
$(a,w) \mapsto  wa$, and $* \mapsto \varepsilon$.
We let $\beta$ append $a$ to the right of a word $w$, which reflects how traces are formed in an LTS.
\item $\Kl(\Pow)$ is left strictly $\DCPOb$-enriched w.r.t.\ the partial order given by pointwise inclusion ($f\sqsubseteq g \colon X\kleislitow Y$ iff $f(x)\seq g(x)$ for all $x\in X$).
\item $\barB$ is locally monotone as the law $\delta^B$ is monotone: $S\seq S'$ implies $\delta^B(a,S)\seq \delta^B(a,S')$.   
\end{enumerate}

The trace map from \Cref{ex:lts-trace}, corresponding to completed trace semantics of LTS, coincides with the coalgebraic trace map \eqref{eq:trace-map}. Indeed, one readily verifies that completed trace semantics forms a $\barB$-coalgebra morphism from $(C,c)$ to the final coalgebra $(\mu B, J\beta^{-1})$.
\end{example}

\section{Abstract GSOS, Revisited}\label{sec:abstract-gsos}

Turi and Plotkin's~\cite{DBLP:conf/lics/TuriP97} abstract GSOS models GSOS-style rule formats categorically and features a free congruence result. We recall
their framework, and at the same time slightly refine it with respect to the required naturality conditions -- a step that enables its application to coalgebraic trace semantics. Abstract GSOS is parametric in the following data:
\begin{itemize}
  \item a base category $\C$ with binary products;
  \item an endofunctor $\Sigma$ on $\C$ with an initial algebra $(\mu\Sigma,\ini)$ and free algebras $(\Sigmas X,\ini_X)$ (so that $\Sigma$ generates a free monad $\Sigmas$);
  \item an endofunctor $B$ on $\C$ with a final coalgebra $(\nu B,\tau)$.
\end{itemize}
Here, the functors $\Sigma$ and $B$ model the \emph{syntax} and \emph{behaviour} of a type of programs or processes under considerations, $\mS$ is the object of closed process terms, and $\nu B$ models abstract behaviours (e.g.\ tree unravellings or traces). The operational semantics of processes is captured by a (pre-)natural transformation that distributes syntax over behaviour:

\begin{defn}
 A \emph{(pre-)GSOS law} of $\Sigma$ over $B$ is (pre-)natural transformation of type
\begin{equation}\label{eq:pre-gsos-law} \rho_X\colon \Sigma(X\times BX)\to B\Sigmas X \qquad (X\in \C).\end{equation}
The \emph{operational model} of \eqref{eq:pre-gsos-law} is defined recursively~\cite[Prop.~2.4.6]{DBLP:books/cu/J2016} as the unique coalgebra\\[-.5ex]
\begin{minipage}{.48\textwidth}
  $\gamma\colon \mu\Sigma\to B(\mu\Sigma)$ such that the diagram on the right commutes. We denote the unique coalgebra morphism into the final coalgebra by
$\beh\colon (\mS,\gamma)\to (\nu B,\tau)$.
\end{minipage}
\begin{minipage}{.5\textwidth}
  \vspace{-.2cm}
\[
\begin{tikzcd}[column sep=15]
 \mS \ar[dashed]{d}[swap]{\gamma} && \Sigma(\mS) \ar{ll}[swap]{\ini}  \ar{d}{\Sigma\langle \id,\,\gamma\rangle} \\
 B(\mS) & B\Sigmas(\mS) \ar{l}[swap]{B\hat\ini} & \Sigma(\mS\times B(\mS)) \ar{l}[swap]{\rho}
\end{tikzcd}
\]
\end{minipage}
\end{defn}
Intuitively, a pre-GSOS law $\rho$ encodes the operational rules of a given process language: for every constructor~$\f$, it specifies the one-step behaviour of a process $\f(p_1,\cdots,p_n)$, i.e.\ the $\Sigma$-terms it transitions into next, depending on the one-step behaviours of its subprocesses $p_i$. The operational model $\gamma$ is the transition system that runs processes according to those rules, and the map $\beh$ sends every process to its abstract behaviour, e.g.\ its bisimilarity class or its set of traces, depending on the structure of the final coalgebra $\nu B$ in the given setting.

\begin{example}\label{ex:gsos}
Consider a process algebra with a parallel composition operator $(p,q)\mapsto p\parallel q$ specified by the operational rules \eqref{eq:rules-par}, where $a$ ranges over a fixed set $L$ of action labels:\\[-1ex]
\begin{minipage}{.45\textwidth}
\begin{equation}\label{eq:rules-par} \frac{p\xto{a} p'}{p \parallel q\xto{a} p' \parallel q} \;\;\;\; \frac{q\xto{a} q'}{p \parallel q\xto{a} p \parallel q'} 
\end{equation}
\end{minipage}
\begin{minipage}{.54\textwidth}
\begin{equation}\label{eq:gsos-rule}
\inference{ x_{i_j}\xto{a_j} y_j\; (j\in J)\;\;\;\; x_{i_k}\centernot{\xrightarrow{b_k}} \;(k\in K)}{\f(x_1,\ldots,x_n)\xto{a} t}
\end{equation}
\end{minipage}\\[2ex]
We take $\Sigma X=X\times X$, $BX=\Pow(L\times X)$ and devise a corresponding GSOS law from~\eqref{eq:rules-par}:
\begin{align*}
\rho_X&\colon (X\times \Pow (L\times X)) \times (X \times \Pow (L\times X)) \to \Pow(L\times \Sigmas X) \qquad (X\in \Set),\\
\rho_X&(p,S,q,T) = \{ (a, p'{\parallel} q) : (a,p')\in S \} \cup \{ (a,p{\parallel} q') : (a,q') \in T\}.
\end{align*}
Naturality of $\rho$ (w.r.t.\ functions, i.e.\ renamings of variables) amounts to the rules \eqref{eq:rules-par} being parametrically polymorphic, that is, they do not inspect the structure of the variables.  

The above translation works more generally for \emph{GSOS rules}~\cite{DBLP:journals/jacm/BloomIM95}, viz.\ rules of the form
\eqref{eq:gsos-rule} where $\f$ is an $n$-ary operation symbol from a given signature $\Sigma$, $i_j, i_k\in \{1,\ldots,n\}$, the variables $x_i$ and $y_j$ are pairwise distinct, $a,a_j,b_k\in L$, and $t$ is a $\Sigma$-term in the variables $x_i$ and $y_j$.
Turi and Plotkin \cite{DBLP:conf/lics/TuriP97} observed that every set $\mathcal{R}$ of GSOS rules can be presented as a GSOS law of the polynomial functor $\Sigma$ over $B$. The operational model of the law is the LTS
\begin{equation}\label{eq:gsos-law-model} \gamma\colon \mS \to \Pow(L\times \mS)  \end{equation}
 whose transitions are inductively determined by the given rules; that is, there is a transition $\f(t_1,\ldots,t_n)\xto{a} s$ 
  iff there exists a rule \eqref{eq:gsos-rule} in $\mathcal{R}$ and terms $s_j\in \mS$ ($j\in J$) such that $t_{i_j}\xto{a_{j}} s_j$ 
   for $j\in J$, $t_{i_k}\centernot{\xrightarrow{b_k}}$ (i.e.\ there is no $b_k$-labelled transition from $t_{i_k}$) for $k\in K$, and the term $s$ arises from $t$ via the substitution $x_i\mapsto t_i$ and $y_j\mapsto s_j$.
 \end{example} 
In what follows, we gloss over the minor technical issue of final coalgebras not existing for functors involving the full power set functor $\Pow$; this is easily remedied 
by replacing $\Pow$ with a bounded subfunctor $\Pow_\alpha X = \{Y\seq X : |Y|<\alpha\}$ for some regular cardinal $\alpha$.
    
The main feature of abstract GSOS is its compositionality result: behavioural equivalence of processes with respect to the final coalgebra $\nu B$ is preserved by all program       
contexts. To make this statement formal, we use the following categorical concept:
\begin{defn}[Morphic Congruence]
A \emph{morphic congruence} on a $\Sigma$-algebra $(A,a)$ is a morphism $h\colon A\to  A'$ of $\C$ that carries a $\Sigma$-algebra morphism, that is, there exists a $\Sigma$-algebra structure $(A',a')$ on $A'$ such that $h\colon (A,a)\to (A',a')$ is a morphism of $\Sigma$-algebras.  
\end{defn}

\begin{rem}\label{rem:congruence-relation}
For a polynomial set functor $\Sigma$, if $h\colon A\to A'$ is a morphic congruence then its kernel $\equiv\,\seq A\times A$ given by $a \equiv a'$ iff $h(a)=h(a')$ forms a congruence relation on~$A$ in the usual sense: for all $n$-ary $\f$ in $\Sigma$, if $a_i \equiv a_i'$ for $i=1,\ldots,n$ then $\f^A(a_1,\ldots,a_n) \equiv \f^A(a_1',\ldots,a_n')$. 

\end{rem}

\begin{theorem}[Congruence Theorem for Pre-GSOS Laws]\label{thm:cong}
Let $\rho$ be a pre-GSOS law that is natural w.r.t.\ the morphisms $\hat\ini\colon \Sigmas(\mS)\to \mS$ and $\Sigmas\beh\colon \Sigmas(\mS)\to \Sigmas(\nu B)$. Then the behaviour morphism $\beh\colon \mS\to \nu B$ is a morphic congruence on the initial algebra $(\mS,\ini)$.  
\end{theorem}

\noindent This theorem is a refined version of Turi and Plotkin's congruence result for GSOS laws~\cite{DBLP:conf/lics/TuriP97}. 
Their original proof involves liftings of comonads to categories of algebras. It is conceptually elegant, but inherently reliant on naturality and thus not applicable to possibly non-natural pre-GSOS laws. Our proof of \Cref{thm:cong} (see appendix) therefore takes a more direct route that allows to analyse the required naturality conditions and make them explicit.
\begin{example}\label{ex:gsos-cong-bisim}
When applied to the setting of \Cref{ex:gsos}, the above theorem implies the congruence result for the GSOS format by Bloom et al.~\cite{DBLP:journals/jacm/BloomIM95}: for every set of GSOS rules, behavioural equivalence on its operational model \eqref{eq:gsos-law-model} (i.e.\ the equivalence relation $\equiv\,\seq \mS\times \mS$ given by $t \equiv s$ iff $\beh(t)=\beh(s)$) is a congruence w.r.t.\ all constructors from~$\Sigma$.
Note that behavioural equivalence coincides with \emph{strong bisimilarity} for LTS~\cite{DBLP:journals/tcs/Rutten00}.
\end{example}

\section{De Simone Laws and Compositional Coalgebraic Trace Semantics}
Our aim is to use \Cref{thm:cong} to reason about congruence properties of coalgebraic trace semantics (\Cref{sec:coalgebraic-trace-semantics}). For this purpose, we introduce \emph{De Simone laws}, a  
subclass of GSOS laws designed to smoothly extend to Kleisli categories. Here the behaviour map corresponds to the trace map, making the theory of the previous section applicable to
trace semantics.
\subsection{De Simone Rules}
To motivate our approach, we consider \emph{De Simone
  rules}~\cite{DeSimone1985}, a type of GSOS rules that is known to
guarantee congruence of trace semantics~\cite{bloom94} rather than just strong bisimilarity.

\begin{defn}[De Simone Rule]\label{def:de-simone-rule} 
Fix a signature $\Sigma$, a set $L$ of labels, and infinitely many variables 
$x_i$ and $y_i$ ($i\in \Nat$). A \emph{De Simone rule} is a GSOS rule of the form
\begin{equation}\label{eq:de-simone-rule}
\frac{x_{i_1}\xto{a_1} y_{i_1} \quad\cdots\quad x_{i_k}\xto{a_k} y_{i_k}}
{\f(x_1,\ldots,x_n)\xto{a} t} 
\end{equation}
where (i)  $\f\in \Sigma$ is an $n$-ary operation, (ii) $i_1,\ldots,i_k\in \{1,\ldots,n\}$ are pairwise distinct, (iii) $a_1,\ldots,a_k,a\in L$, and (iv) $t$ is an affine $\Sigma$-term in the variables $y_i$ ($i\in \{i_1,\ldots,i_k\}$) and $x_j$ ($j\in\{1,\ldots,n\}\setminus\{i_1,\ldots,i_k\}$). (A term is \emph{affine} if no variable occurs more than once in it.)
\end{defn}

\begin{example}
The two rules \eqref{eq:rules-par} of parallel composition are De Simone rules up to renaming the variables $p$, $p'$, $q$, $q'$ to $x_1$, $y_1$, $x_2$, $y_2$.
\end{example}

\noindent Since De Simone rules are GSOS rules, they can be presented as GSOS laws of~$\Sigma$ over $BX=\Pow(L\times X)$, viz.\ natural transformations
$\rho_X\colon \Sigma(X\times \Pow(L\times X)) \to \Pow(L\times \Sigmas X)$.
However, the specific form of De Simone rules allows for a simpler kind of natural transformation:

\begin{construction}\label{rk:star}
Every set $\R$ of De Simone rules induces a natural transformation 
\begin{equation}\label{eq:de-simone-law-from-rules}\rho_X\colon \Sigma(X+L\times X+1)\to \Pow(L\times \Sigmas X + 1) \qquad (X\in \Set)\end{equation} 
where $\rho_X$ maps $\f(u_1,\ldots,u_n)\in \Sigma(X+L\times X+1)$ to the following subset of $L\times \Sigmas X+1$:
\begin{itemize}
\item If $u_i=*$ for some $i$, then $\rho_X(\f(u_1,\ldots,u_n))=\{*\}$. 
\item Otherwise let $\{i_1\ldots,i_k \} = \{i\mid u_i\in L\times X\}$; thus $u_{i_l}=(a_l,v_{i_l})$ for some $a_l\in L$ and $v_{i_l}\in X$, and $u_j\in X$ for $j\not\in \{i_1,\ldots,i_k\}$. Then $\rho_X(\f(u_1,\ldots,u_n))$ consists of $*$ and all pairs $(a,t[\sigma])\in L\times \Sigmas X$ such that $\R$ contains the rule \eqref{eq:de-simone-rule} and $t[\sigma]$ emerges from the term $t$ via the substitution $\sigma$ given by $y_{i_l}\mapsto v_{i_l}$ and $x_j \mapsto u_j$ for $j\not\in \{i_1,\ldots,i_k\}$.
\end{itemize}
Note that the `$+1$' component flags all states as terminating, as always $*\in \rho_X(\f(u_1,\ldots,u_n))$.
This is necessary to fit the coalgebraic trace semantics framework (cf.\ \Cref{ex:lts-trace}), and makes sure that all 
emerging partial traces are being tracked.
\end{construction}

\subsection{De Simone Laws}
The above modelling of De Simone rules as natural transformations of type~\eqref{eq:de-simone-law-from-rules} immediately suggests a categorical generalization. In the remainder, let us fix the following data:
\begin{itemize}
\item a category $\C$ with finite limits and finite coproducts;
\item two endofunctors $\Sigma$, $B$ and a monad $T$ on $\C$, where $\Sigma$ has an initial algebra $(\mu\Sigma,\ini)$ and free algebras $\Sigmas X$.
\end{itemize}
\begin{defn}[Pre-De Simone Law]
A \emph{pre-De Simone law}  is\,a\,natural\,transformation
\begin{equation}\label{eq:de-simone-law} \rho_X\colon \Sigma(X+BX)\to TB\Sigmas X \qquad (X\in \C). \end{equation}
\end{defn}
De Simone rules thus correspond to pre-De Simone laws with $\C=\Set$, $BX=L\times X + 1$, $T=\Pow$, and a polynomial functor $\Sigma$.
To reason about abstract congruence properties for trace semantics of pre-De Simone laws, the given data needs some additional structure:

\begin{assumptions}\label{asm}
In the following, we assume that:
\begin{enumerate}
\item\label{asm1} There is a natural isomorphism 
$j\colon TX\times TY\xto{\cong} T(X+Y)$ ($X,Y\in \C$).
\item \label{asm2} There are functor-over-monad distributive laws of $\Sigma$ and $B$ over $T$. We denote the laws and their Kleisli extensions by 
$\delta^{\Sigma}\colon \Sigma T\to T\Sigma$, $\delta^{B}\colon BT\to TB$, and $\barSigma,\barB\colon \Kl(T)\to \Kl(T)$.    
\item\label{asm4} $T$, $B$, $\barB$ form a setting for coalgebraic trace semantics (\Cref{def:trace-setting}).
\end{enumerate}
\end{assumptions}

\begin{rem}\label{rem:products}  
Condition \ref{asm1} entails that $\Kl(T)$ has binary products, which coincide with binary coproducts in $\Kl(T)$ and thus binary coproducts in $\C$~\cite[Thm.~19]{cj13}. Monads~$T$ satisfying~\ref{asm1} and additionally $T0\cong 1$ are known as \emph{additive monads}~\cite{cj13}. 
\end{rem}

\begin{rem}\label{rem:sigmas-t-dist}
The distributive law $\delta^{B}$ extends to a functor-over-monad distributive law of $B_0 X = X+BX$ over $T$ and thus a corresponding Kleisli extension of $B_0$:
\[
\delta^{B_0} = (\,B_0 TX\xto{\id+\delta^{B}} TX+TBX \xto{[T\inl,T\inr]} TB_0X \,)\qqand \barB_0\colon \Kl(T)\to \Kl(T).\]
\end{rem}

\begin{rem}\label{rem:sigmas-t-dist}
By \Cref{prop:initial-alg-lift}, the free monad $\barSigmas$ of $\barSigma$ is a Kleisli extension of the free monad $\Sigmas$ of $\Sigma$. We write $\delta^{\Sigmas}\colon \Sigmas T\to T\Sigmas$ for the ensuing distributive law extending $\delta^{\Sigma}$.
\end{rem}
Since products coincide with coproducts in $\Kl(T)$ (\Cref{rem:products}), every pre-De Simone law~\eqref{eq:de-simone-law} forms a pre-GSOS law of $\barSigma$ over $\barB$ in $\Kl(T)$, that is, a family of morphisms
\begin{equation}\label{eq:gsos-law-kleisli} 
\rho_X\colon \barSigma(X\times\barB X) \kleislitow \barB\barSigmas X \qquad (X\in \Kl(T)).
\end{equation}
Additionally, a pre-De Simone law yields a GSOS law in $\C$ with the same operational model:
\begin{proposition}\label{lem:op-model-coincidence}
Every pre-De Simone law \eqref{eq:de-simone-law} induces a GSOS law of $\Sigma$ over $TB$ in $\C$:
 \begin{equation}\label{eq:bar-rho} \begin{tikzcd}
 {\Sigma(X\times TBX)} &&[1em]& {TB\Sigmas X} \\
 {\Sigma(TX\times TBX)} &
 {\Sigma T(X+ BX)} &
 {T\Sigma(X+ BX)} & {TTB\Sigmas X}
  \arrow["\barrho_X"', swap, from=1-1, to=1-4]
  \arrow["{\Sigma(\eta_X\times\id)}"', from=1-1, to=2-1]
  \arrow["{\Sigma j}"', swap, from=2-1, to=2-2]
  \arrow["{\delta^{\Sigma}_{X+BX}}"', swap, from=2-2, to=2-3]
  \arrow["{T\rho_X}", from=2-3, to=2-4]
  \arrow["\mu_{B\Sigmas X}"', from=2-4, to=1-4]
 \end{tikzcd}
 \end{equation}
 The laws $\rho$ \eqref{eq:gsos-law-kleisli} and $\barrho$ \eqref{eq:bar-rho} have the same operational model $\gamma\c\mS\to TB(\mS)$.
 \end{proposition}
Note that \eqref{eq:gsos-law-kleisli} is generally not a GSOS law: naturality of \eqref{eq:de-simone-law} in $\C$ corresponds to naturality of \eqref{eq:gsos-law-kleisli} w.r.t.~\emph{pure} morphisms, but not arbitrary morphisms in~$\Kl(T)$. To apply the general congruence result for pre-GSOS laws (\Cref{thm:cong}) to the pre-De Simone law \eqref{eq:gsos-law-kleisli}, we need additional naturality conditions beyond pure morphisms. We use the following terminology:

\begin{defn}\label{def:natural}
A pre-De Simone law \eqref{eq:de-simone-law} is \emph{natural} w.r.t.\ a morphism $p\colon P\kleislito X$ in $\Kl(T)$ if the corresponding pre-GSOS law \eqref{eq:gsos-law-kleisli} in $\Kl(T)$ is natural w.r.t.~$p$. 
\end{defn}

\begin{proposition}\label{prop:affine-vs-natural}
A pre-De Simone law \eqref{eq:de-simone-law} is natural w.r.t.~$p\colon P\kleislito X$ iff the rectangle in $\C$ shown below commutes when precomposed with $\Sigma B_0 p$ and postcomposed with $\mu$: \begin{equation}\label{eq:affine}
\begin{tikzcd}
  \Sigma B_0 P \ar{r}{\Sigma B_0 p} 
  & \Sigma B_0 TX \ar{r}{\Sigma\delta^{B_0}} \ar{d}[swap]{\rho_{TX}} & \Sigma TB_0 X \ar{r}{\delta^{\Sigma}} & T\Sigma B_0 X \ar{d}{T\rho_X} \\
  & TB\Sigmas TX \ar{r}{TB\delta^{\Sigmas}} & TBT\Sigmas X \ar{r}{T\delta^{B}} & TTB\Sigmas X \ar{r}{\mu} & TB\Sigmas X
\end{tikzcd}
\end{equation}
\end{proposition}

\subsection{Compositionality}
For every pre-De Simone law \eqref{eq:de-simone-law} with its corresponding pre-GSOS law \eqref{eq:gsos-law-kleisli} in $\Kl(T)$, let us denote the operational model and its associated trace (i.e.\ behaviour) morphism by 
$\gamma\colon \mu\barSigma \kleislito \barB(\mu\barSigma)$ and $\tr\colon \mu\barSigma\kleislito \nu \barB = \mu \barB$. Our congruence result for this setting (\Cref{thm:cong-trace-positive} below) involves another monad $T_{+}$ alongside $T$, thought of as the `positive' part of $T$. For example, for $T=\Pow$ one takes the submonad $T_{+}=\Pow_{+}$, the non-empty power set monad. This is an instance of a general construction for 
computing the \emph{affine part} of a monad. (A monad $T$ is \emph{affine} if $T1\cong 1$). Recall that we assume $\C$ to have finite limits.
\begin{notheorembrackets}
\begin{proposition}[{\cite{Lindner1979-qe}}]\label{pro:affine-part}
For every $X\in \C$, let $i_X\c T_{+}X\to TX$ be the equalizer
\begin{equation*}
\begin{tikzcd}[column sep=12ex, row sep=normal]
T_{+}X
	\rar["i_X"] &
TX
	\ar[r,shift right=.35ex,"T!"']
	\ar[r,shift left=.75ex,"\eta\,\comp\,!"] &
T1.
\end{tikzcd}
\end{equation*}
Then $T_{+}$ extends to an affine monad and the family $i_X$ ($X\in \C$) to a monad 
morphism.
\end{proposition}
\end{notheorembrackets}
Let us call a morphism $p\c P\to TX$ \emph{affine} if it factors through $i_X\c T_{+}X\to TX$. From the universal property of the equalizer $i_{X}$, the following is immediate:

\begin{lemma}\label{lem:positive-criterion}
A morphism $p\c P\to TX$ is affine iff $T!\comp p = \eta\,\comp\, !$.
\end{lemma}

\begin{defn}[De Simone Law]
A \emph{De Simone law} is a pre-De Simone law \eqref{eq:de-simone-law} that is natural w.r.t.~the $\Kl(T)$-morphism $i_X\c T_{+}X\kleislito X$ for each $X\in \C$ (cf.\ \Cref{def:natural}).
\end{defn}
This leads to our main result, a general compositionality criterion for De Simone laws:
\begin{theorem}[Compositionality of Coalgebraic Trace Semantics]\label{thm:cong-trace-positive}
For every De Simone law \eqref{eq:de-simone-law} whose trace morphism $\tr\colon \mS\to T(\mu B)$ is affine, $\tr$ is a morphic $\Sigma$-congruence.
\end{theorem}
The proof (see appendix) uses the congruence result for pre-GSOS laws (\Cref{thm:cong}) and amounts to verifying that the naturality conditions of the latter follow from the naturality of the De Simone law \eqref{eq:de-simone-law} w.r.t.~$i_X\colon T_{+}X\kleislito X$ and the affineness of $\tr$. 

We shall see below that the naturality properties of \eqref{eq:de-simone-law} reflect precisely the requirement of affineness of output terms in De Simone rules (\Cref{def:de-simone-rule}). In contrast, affineness of the trace morphism~$\tr$ is a non-trivial global property. For example, in a setting of probabilistic specifications (\Cref{sec:prob-de-simone}) it expresses that the operational model is \emph{almost surely terminating}. On a technical level, by \Cref{lem:positive-criterion}, affineness of $\tr$ amounts to the 
equality $T!\comp\tr = \eta\,\comp\, !$. It turns out that one of the inequalities constituting 
it follows nearly for free:
\begin{proposition}\label{prop:tr-ineq}
If the morphism $1\xto{\langle \eta,\eta \rangle} T1\times T1\xto{j} T(1+1)$ and every component~\eqref{eq:de-simone-law} is affine, then so is $\gamma$, and moreover the induced trace morphism $\tr$
satisfies $T!\comp\tr \sqsubseteq \eta\,\comp\, !$.
\end{proposition}

\section{Case Study I: De Simone Rules}\label{sec:de-simone}

As a first showcase of our abstract theory, we derive the congruence of De Simone rules.

\subsection{Categorical Setting}
Recall that De Simone rules yield pre-De Simone laws for $\C=\Set$, $ BX=L\times X + 1$, $T=\Pow$, and a polynomial functor $\Sigma$. 
Let us verify that this data satisfies the \Cref{asm}:
\begin{enumerate}
\item We have the natural isomorphism 
	$j\colon \Pow X\times \Pow Y\cong \Pow (X+Y)$ given by $j(S,T) = S \uplus T$.
\item The distributive law  $\delta^{B}\colon L\times \Pow X+1\to \Pow(L\times X+1)$ is given by \Cref{ex:lts-trace-coalg}, and the law $\delta^{\Sigma}\colon \Sigma\Pow\to \Pow\Sigma$ by
$\f(X_1,\ldots,X_n) \;\mapsto\; \{\, \f(x_1,\ldots,x_n) \mid x_i\in X_i \text{ for $i=1,\ldots,n$}\,\}$.
\item $\Pow,B,\barB$ form a setting for coalgebraic trace semantics, as explained in \Cref{ex:lts-trace-coalg}.
\end{enumerate}
The affine part of $\Pow$ is the non-empty power set monad $\Pow_{+}$.
Thus a morphism $p\colon P\kleislito X$ of $\Kl(\Pow)$ is affine iff $p(z)\neq \emptyset$ for all $z\in P$, that is, $p$ corresponds to a right-total relation.

\subsection{Compositionality}
The key lies in the next proposition, whose proof illustrates how the affineness of terms in De Simone rules corresponds to the induced pre-De Simone law being a (proper) De Simone law:

\begin{proposition}\label{prop:de-simone-rule-law-affine}
For every set of De Simone rules, \eqref{eq:de-simone-law-from-rules} is a De Simone law.
\end{proposition}

\begin{proof}[Proof sketch]
By \Cref{prop:affine-vs-natural}, we are to show that the two legs of the diagram below send any element $\f(U_1,\ldots,U_n)$ of $\Sigma(\Pow_{+} X+L\times \Pow_{+} X + 1)$ to the same element of $\Pow(L\times \Sigmas X+1)$.
\[
\begin{tikzcd}[column sep=2ex, scale cd=.95]
\Sigma(\Pow X+L\times \Pow X + 1) \ar{r}{\Sigma\delta^{B_0}} \ar{d}[swap]{\rho_{\Pow X}} &[2ex] \Sigma \Pow (X+L\times X + 1) \ar{r}{\delta^{\Sigma}}  &[.3ex] \Pow \Sigma (X+L\times X + 1) \ar{d}{\Pow \rho_X} &[-.2ex] \\
\Pow(L\times \Sigmas \Pow X +1) \ar{r}{\Pow B\delta^{\Sigmas}} & \Pow (L\times \Pow \Sigmas X +1) \ar{r}{\Pow \delta^{B}} & \Pow \Pow (L\times \Sigmas X+1) \ar{r}{\mu} & \Pow(L\times \Sigmas X + 1)
\end{tikzcd}
\]
To illustrate the gist of the argument, let us consider a binary operator $\f$ and assume that the only rule for $\f$ is given by \eqref{eq:proof-example-de-simone-rule} below, for fixed labels $a,a_1\in L$:\\[-.5ex]
\begin{minipage}{.35\textwidth}
\begin{equation}\label{eq:proof-example-de-simone-rule}
\frac{x_{1}\xto{a_1} y_{1}}
{\f(x_1,x_2)\xto{a} \f(y_1,x_2)}
\end{equation}
\end{minipage}
\begin{minipage}{.1\textwidth}
~
\end{minipage}
\begin{minipage}{.35\textwidth}
\begin{equation}\label{eq:proof-example-de-simone-rule-non-affine}
\frac{x_{1}\xto{a_1} y_{1}}
{\f(x_1,x_2)\xto{a} \f(y_1,y_1)}
\end{equation}
\end{minipage}\\[2ex]
The interesting case is that of an element $\f(U_1,U_2)$ with $U_1=(a_1,V_1)\in L\times \Pow_{+}X$ and $U_2\in \Pow_{+}X$.   
 The map $\rho_{\Pow X}$ sends $\f(U_1,U_2)$ to the set 
 $\{ *, (a,\f(V_1,U_2))\} \in \Pow(L\times \Sigmas\Pow X+1)$.  
 The pair $(a,\f(V_1,U_2))$ is mapped by $\delta^{B}\comp B\delta^{\Sigmas}$ to the set of all pairs
$(a,\f(v_1,u_2))$ where $v_1\in V_1$ and $u_2\in U_2$. (Here affineness of the term $\f(y_1,x_2)$ in \eqref{eq:proof-example-de-simone-rule} is used, see \Cref{rem:affine}.) Thus, overall,
$\mu\comp \Pow\delta^{B}\comp \Pow B\delta^{\Sigmas}\comp \rho_{\Pow X}(\f(U_1,U_2))$
is the subset of $L\times \Sigmas X + 1$ consisting of:
\[ \text{(1) the element $*$ of $1$;}\qquad \text{(2) all pairs $(a,\f(v_1,u_2))\in L\times \Sigmas X$ with $v_1\in V_1$ and $u_2\in U_2$.} \] 
Similarly, $\f(U_1,U_2)$ is mapped by $\delta^{\Sigma}\comp \Sigma\delta^{B_0}$ to the non-empty set of all $\f((a_1,v_1),u_2)$ with $v_1\in V_1$ and $u_2\in U_2$. Moreover, $\rho_X$ maps each such $\f((a_1,v_1),u_2)$ to the set $\{*,(a,\f(v_1,u_2))\}$. Thus
$\mu\comp  \Pow\rho_X\comp \delta^{\Sigma}\comp \Sigma\delta^{B_0}(\f(U_1,U_2))$ is also given by the above elements (1) and (2).
\end{proof}

\begin{rem}\label{rem:affine} 
Some subtleties of the above proof are worth mentioning:
\begin{enumerate}
\item The reasoning critically rests on affineness of the output term $\f(y_1,x_2)$ of the rule \eqref{eq:proof-example-de-simone-rule} and fails for non-affine terms. For example, suppose that $\f$ is instead specified by the rule~\eqref{eq:proof-example-de-simone-rule-non-affine}
with the output term $\f(y_1,y_1)$, which is non-affine as the variable $y_1$ occurs twice. Now
\begin{itemize}\item $\mu\comp \Pow\delta^{B}\comp \Pow B\delta^{\Sigmas}\comp \rho_{\Pow X}(\f(U_1,U_2))$    
consists of $*$ and all $\f(v_1,v_1')$ with $v_1,v_1'\in V_1$. 
\item $\mu\comp  \Pow\rho\comp \delta^{\Sigma}\comp \Sigma\delta^{B_0}(\f(U_1,U_2))$
consists of $*$ and all $\f(v_1,v_1)$ with $v_1\in V_1$.
\end{itemize} 
These two sets differ whenever $V_1$ has more than one element.
\item It is also critical that the sets $V_1$ and $U_2$ considered in the proof are \emph{non-empty}. For example, suppose that $V_1=\emptyset$ (so $U_1=(a_1,\emptyset)$). Then
\[\mu\comp \Pow\delta^{B}\comp \Pow B\delta^{\Sigmas}\comp \rho(\f(U_1,U_2)) = \{*\} \qquad\text{while}\qquad  
\mu\comp  \Pow\rho\comp \delta^{\Sigma}\comp \Sigma\delta^{B_0}(\f(U_1,U_2)) = \emptyset.\]
Recall that we need `$*$' to signal potential trace termination (\Cref{rk:star}).
\end{enumerate}
\end{rem}
From the above proposition the congruence result for De Simone rules is now immediate. Let $\mathcal{R}$ be a set of such rules with operational model $\gamma\colon \mu\Sigma\to \Pow(L\times \mS+1)$ and trace map $\tr\colon \mS\to \Pow(\mu B) = \Pow(L^{*})$.
Two terms $t,s\in \mu\Sigma$ are \emph{trace equivalent} if $\tr(t)=\tr(s)$.

\begin{theorem}[Compositionality of De Simone Rules~\cite{bloom94}]\label{thm:cong-de-simone-rules}
For every set of De Simone rules, trace equivalence on the operational model is a congruence relation.
\end{theorem}

\begin{proof}
The trace map $\tr$ is affine since $\epsilon\in \tr(t)$ for all $t$, so by \Cref{prop:de-simone-rule-law-affine} and \Cref{thm:cong-trace-positive} it forms a morphic congruence. By \Cref{rem:congruence-relation}, this implies that trace equivalence (the kernel of $\tr$) forms a congruence relation.
\end{proof}

\section{Case Study II: Probabilistic De Simone Rules}\label{sec:prob-de-simone}
Thanks to the categorical generality of De Simone laws, varying the underlying monad $T$ leads to novel De Simone-type congruence formats for LTS with computational effects that are more complex than non-determinism. For illustration, we move to a probabilistic setting.

\subsection{Weighted and Probabilistic Labelled Transition Systems}\label{sec:trace-weighted}
The systems under consideration are a weighted form of LTS, which involve a weighted generalization of the power set monad. 
Consider the endofunctor $\W\c\Set\to\Set$ given by 
\begin{align*}
\W X = [0,\infty]^X,\qquad \W(f\colon X\to Y)\c \sum_{i\in I} r_i\cdot x_i \mapsto \sum_{i\in I} r_i\cdot f(x_i).
\end{align*}
Here a function $\phi\in \W X$ is identified with a formal sum $\sum_{i\in I} r_i\cdot x_i$ where $r_i\in [0,\infty]$, $x_i\in X$, and $\sum_{x_i=x} r_i = \phi(x)$  for all $x\in X$. 
The functor $\W$ extends to a monad with unit $\eta_X\colon X\to \W X$ and multiplication $ \mu_X\colon \W\W X \to \W X$ 
defined as follows:
\begin{align*}
\eta_X\c x\mapsto 1\cdot x,\qquad\mu_X\c \sum_{i\in I} r_i\cdot (\sum_{j\in J_i}r_{ij}\cdot x_{ij} ) \mapsto \sum_{i\in I}\sum_{j\in J_i} r_i\cdot r_{ij}\cdot x_{ij}.
\end{align*}
The affine part $\W_{+}$ is the \emph{distribution monad} $\D X = \{ \phi\in \W X\mid \sum_{x\in X} \phi(x) =1  \}$ formed by discrete probability distributions. Thus a morphism $p\colon P\kleislito X$ in $\Kl(\W)$ is affine iff $p(z)\in \D X$ for all $z\in P$.
A \emph{$\W$-LTS} (with explicit termination) is a coalgebra
\begin{equation}\label{eq:weighted-lts} c\colon C\to \W (L\times C+1). \end{equation} 
We write $x\xto{a/r} y$ if $c(x)(a,y)=r$, and $x\xto{r}*$ if $c(x)(*)=r$. A \emph{(generative) probabilistic LTS}~\cite{hjs07} is a $\W$-LTS whose structure $c$ is affine ($c(x)\in \D(L\times C+1)$ for $x\in C$). Here, weights are interpreted as probabilities: $x\xto{a/r} y$ means that with probability $r$ the process~$x$ outputs~$a$ and transitions into $y$, and $x\xto{r}*$ means that $x$ terminates with probability $r$. 
\begin{rem}
We stick to $[0,\infty]$ for simplicity. All present definitions and results extend to weights from any continuous commutative semiring~\cite{dkv09}. In particular, taking the boolean semiring $\{0,1\}$ (with $\vee$ as addition and $\wedge$ as multiplication) recovers the results of \Cref{sec:de-simone}.
\end{rem}
The \emph{trace map} of a $\W$-LTS \eqref{eq:weighted-lts} is the map $\tr\colon C\to \W(L^{*})$ given by
\begin{equation}\label{eq:trace-weighted} 
 \tr(x)(a_1\cdots a_n) = \sum_{x=x_0,\ldots,x_n} \big(\prod_{i=0}^{n-1} c(x_i)(a_{i+1},x_{i+1})\big)\cdot c(x_n)(*),
\end{equation}
with the sum ranging over all paths of length $n$ starting from $x$, that is, all lists $x_0,\ldots,x_n$ of elements of $C$ where $x=x_0$. In the probabilistic case, $\tr(x)(a_1\cdots a_n) $ is the probability that the process $x$ generates the trace $a_1\cdots a_n$ and then terminates. This definition is captured by the setting $T$, $B$, $\barB$ for coalgebraic trace semantics where
$T=\W$ and $BX=L\times X+1$,
and the Kleisli extension
$\barB\colon \Kl(\W)\to \Kl(\W)$
corresponds to the distributive law 
\begin{equation*} \delta^{B}\colon L\times \W X + 1 \to \W(L\times X+1),\qquad  (a,\sum_{i\in I} r_i\cdot x_i) \mapsto \sum_{i\in I} r_i\cdot (a,x_i),\qquad *\mapsto 1\cdot *. \end{equation*}
Thus $\barB$ sends $f\colon X\kleislito Y$ to $\barB f\colon L\times X+1\kleislito L\times Y+1$ given by
\begin{equation}\label{eq:barB-weighted} (a,x)\mapsto \sum_{i\in I} r_i\cdot (a,y_i) \qqand *\mapsto 1\cdot *,\end{equation}
where $f(x)=\sum_{i\in I} r_i\cdot y_i$. Then all conditions of \Cref{def:trace-setting} are satisfied:
\begin{enumerate}
\item $B$ has the initial algebra $(\mu B,\beta)$ carried by $L^{*}$, as before.
\item The category $\Kl(\W)$ is left strictly $\DCPOb$-enriched: for $f,g\colon X\kleislito Y$ the order is given by 
$f\sqsubseteq g$ iff $f(x)\leq g(x)$ for all $x\in X$, where $\leq$ is the usual order of the interval $[0,\infty]$.
\item The lifting $\barB$ is locally monotone; this is clear from \eqref{eq:barB-weighted}. 
\end{enumerate}
For every $\W$-LTS \eqref{eq:weighted-lts}, the coalgebraic trace map $\tr\colon C\kleislito \nu \bar B = \mu \barB = L^{*}$
is given by \eqref{eq:trace-weighted}. Indeed, one readily verifies that \eqref{eq:trace-weighted} is a coalgebra morphism from $(C,c)$ to $(\mu B, J\beta^{-1})$.

\subsection{A Probabilistic Rule Format}
We now introduce a probabilistic version of De Simone rules. Its syntax and semantics are closely related to \emph{probabilistic GSOS}~\cite{bartels02} and \emph{stochastic GSOS}~\cite{DBLP:conf/fossacs/KlinS08,ks13}. The main difference is that the latter apply to reactive probabilistic systems, while our rules apply to generative ones, which is the more natural setting for trace semantics. In addition, we put restrictions on the premises and the output terms of rules similar to the original De Simone format.

\begin{defn}[Weighted De Simone Rule]\label{de:de-simone-rule-weighted}
Fix a signature $\Sigma$, label set $L$, and variables $x_i$, $y_i$ as in \Cref{def:de-simone-rule}. A \emph{weighted De Simone rule} is a rule of either of the two forms
\begin{align}
\frac{x_{i_1}\xto{a_1} y_{i_1} \quad\cdots\quad x_{i_k}\xto{a_k} y_{i_k} \quad x_{j_1} \to * \quad\cdots\quad x_{j_m} \to *}
{\f(x_1,\ldots,x_n)\xto{a/r} t} \label{eq:de-simone-rule-weighted} \\
\frac{x_{i_1}\xto{a_1} y_{i_1} \quad\cdots\quad x_{i_k}\xto{a_k} y_{i_k} \quad x_{j_1} \to * \quad\cdots\quad x_{j_m} \to *}
{\f(x_1,\ldots,x_n)\xto{r} *} \label{eq:de-simone-rule-weighted-2}
\end{align}
where (i) $\f$ is an $n$-ary operation symbol from $\Sigma$, (ii) the indices $i_1,\ldots,i_k,j_1,\ldots,j_m\in \{1,\ldots,n\}$ are pairwise distinct, (iii) $a_1,\ldots,a_k,a\in L$, (iv) $r\in [0,\infty]$, and (v) $t$ is an affine $\Sigma$-term in the variables $y_i$ ($i\in \{i_1,\ldots,i_k\}$) and $x_j$ ($j\in\{1,\ldots,n\}\setminus \{i_1,\ldots,i_k, j_1,\ldots,j_m\}$).
\end{defn}

\begin{rem}
Weighted De Simone rules \eqref{eq:de-simone-rule-weighted} could be  more suggestively written as 
\begin{align*}
\frac{x_{i_1}\xto{a_1/r_1} y_{i_1} \;\;\cdots\;\; x_{i_k}\xto{a_k/r_k} y_{i_k} \quad x_{j_1} \xto{r_{1}'} * \;\;\cdots\;\; x_{j_m} \xto{r_m'} *}
{\f(x_1,\ldots,x_n)\xto{a/r\cdot r_1\cdot \cdots \cdot r_n\cdot r_1'\cdot \cdots \cdot r_m'} t}
\end{align*}
with variables $r_i,r_j'$ representing weights; similarly for \eqref{eq:de-simone-rule-weighted-2}.  This indicates how the transition weights of the premises contribute to the transition weight in the conclusion. More formally:
\end{rem}

\begin{defn}[Operational Model]
Given a set $\R$ of weighted De Simone rules, the \emph{operational model} of $\R$ is the coalgebra
\begin{equation}\label{eq:op-model-weighted} \gamma\colon \mS\to \W(L\times \mS+1) \end{equation}
defined by structural recursion as follows. Given a term $\f(t_1,\ldots,t_n)$ in $\mS$ and a rule $R\in \R$ for $\f$, we let $\gamma_R(\f(t_1,\ldots,t_n))\in \W(L\times \mS+1)$ denote the weighted sum 
\[ \begin{cases}
\sum_{s_{i_1},\ldots,s_{i_k}\in \mS}\; r\cdot \prod_{l=1}^k \gamma(t_{i_l})(a_{l},s_{i_k}) \cdot \prod_{l=1}^m c(t_{j_l})(*)\cdot (a,t[\sigma]) & \text{if $R$ is given by \eqref{eq:de-simone-rule-weighted}};\\
 \sum_{s_{i_1},\ldots,s_{i_k}\in \mS}\; r\cdot \prod_{l=1}^k \gamma(t_{i_l})(a_{l},s_{i_k}) \cdot \prod_{l=1}^m c(t_{j_l})(*)\cdot * & 
 \text{if $R$ is given by \eqref{eq:de-simone-rule-weighted-2}}.
\end{cases}
\]
Here $\sigma$ denotes the substitution $x_i\mapsto t_i$  and $y_{i_l}\mapsto s_{i_l}$. Then the coalgebra \eqref{eq:op-model-weighted} is given by
$\gamma(\f(t_1,\ldots,t_n)) = \sum_{R\in \R\text{ rule for $\f$}}\; \gamma_R(\f(t_1,\ldots,t_n))$.
\end{defn}

\begin{defn}[Probabilistic De Simone Specification]\label{def:prob-de-simone}
A \emph{probabilistic De Simone specification} is a set of weighted De Simone rules whose model \eqref{eq:op-model-weighted} is a probabilistic LTS. The specification is \emph{almost surely terminating} (\emph{a.s.t.}) if the trace map $\tr$ of \eqref{eq:op-model-weighted} is affine ($\tr(t)\in \D(L^{*})$ for $t\in \mS$, i.e.\ there is no positive probability of generating an infinite trace).
\end{defn}

\begin{example}\label{ex:prob-lts}
Consider a probabilistic process algebra with a terminating process~$\mathsf{nil}$, a prefixing operator $a.(-)$ for each $a\in L$, and a parallel composition operator $\parallel_r$ for each $r\in [0,1]$. The process $p\parallel_r q$ progresses one of the processes $p$ or $q$ with probability $r$ or $1-r$, respectively. This is formalized by the a.s.t.\ probabilistic De Simone specification below:
  \begin{gather*}
    \inference{}{\mathsf{nil}\xto{1}*}
    \qquad
    \inference{}{a.p \xto{a/1} p }
    \qquad
    \inference{p\xto{a} p'}{p\parallel_r q \xto{a/r} p' \parallel_r q} \qquad 
    \inference{p\to *}{p\parallel_r q \xto{r} *} \\[1ex]
    \inference{q\xto{a} q'}{p\parallel_r q \xto{a/1-r} p\parallel_r q'} \qquad 
    \inference{q\to *}{p\parallel_r q \xto{1-r} *} 
  \end{gather*}
A simple structural induction shows that the operational model \eqref{eq:op-model-weighted} is indeed an a.s.t.\ probabilistic LTS (see appendix for details).
\end{example}

\begin{rem}
It would be desirable to identify syntactic criteria for a set of weighted De Simone rules to form a probabilistic specification. Such criteria exist for probabilistic GSOS for reactive systems~\cite{bartels02}, but appear to be harder to come by in the generative case. Note that \Cref{prop:tr-ineq}
would not help with \Cref{ex:prob-lts}: the components of the associated De Simone 
law given in \Cref{cons:weighted-de-simone-law} below are not affine, even though the operational model \eqref{eq:op-model-weighted} of the law itself is affine, that is, a probabilistic LTS. 
\end{rem}

\subsection{Categorical Setting}
To capture weighted De Simone rules within our categorical framework, we take
$\C=\Set$, $BX=L\times X + 1$, $T=\W$, and a polynomial functor $\Sigma$.
This data satisfies the \Cref{asm}:
\begin{enumerate}
\item We have the natural isomorphism
$j\colon \W X\times \W Y \cong \W(X+Y)$          
given by
\[j(\sum_{i\in I} r_i\cdot x_i,\sum_{j\in J} s_j\cdot y_j)= \sum_{i\in I} r_i\cdot x_i + \sum_{j\in J} s_j\cdot y_j.\]
\item The distributive law $\delta^{B}$ is given by \Cref{ex:lts-trace-coalg}, and $\delta^{\Sigma}\colon \Sigma\W X\to \W\Sigma X$ by
\[ \f(\phi_1,\ldots,\phi_n) \;\;\mapsto\;\; \sum_{i_1\in I_1}\cdots \sum_{i_n\in I_n} r^1_{i_1}\cdots r^n_{i_n}\cdot \f(x^1_{i_1},\cdots,x^n_{i_n}) 
\;\;\;\text{where}\;\;\; \phi_j = \sum_{i\in I_j} r_i^j\cdot x_i^j.\] 
\item $\W,B,\barB$ form a setting for coalgebraic trace semantics, as explained in \Cref{sec:trace-weighted}.
\end{enumerate}

\begin{construction}\label{cons:weighted-de-simone-law}
Every set $\R$ of weighted De Simone rules induces a pre-De Simone law of $\Sigma$ over $B$, i.e.\ a natural transformation 
\begin{equation}\label{eq:de-simone-law-from-rules-weighted}\rho_X\colon \Sigma(X+L\times X+1)\to \W (L\times \Sigmas X + 1) \qquad (X\in \Set)\end{equation} 
defined as follows. Given $\f(u_1,\ldots,u_n)\in \Sigma(X+L\times X+1)$, let 
$i_1,\ldots,i_k$ be those $i$ with $u_i\in L\times X$, say $u_{i_l}=(a_l,v_{i_l})$, and let $j_1,\ldots,j_m$ be those $j$ with $u_j=*$.
We say that a rule $R\in \R$ \emph{matches} $\f(u_1,\ldots,u_n)$ if it is of the form \eqref{eq:de-simone-rule-weighted} or \eqref{eq:de-simone-rule-weighted-2}, that is, its premises match the above indices and labels. Let $r_R=r$ be the output weight of the rule. Moreover let $a_R=a$ and $t_R=t$ be the output label and output term if the rule is of type \eqref{eq:de-simone-rule-weighted}. We then put 
\[ \rho_X(\f(u_1,\ldots,u_n)) = \sum_{R~\eqref{eq:de-simone-rule-weighted} } r_R\cdot t_R[\sigma]  + \sum_{R~\eqref{eq:de-simone-rule-weighted-2} } r_R\cdot *\]
where the sums range, respectively, over all $R\in \R$ of type \eqref{eq:de-simone-rule-weighted} or \eqref{eq:de-simone-rule-weighted-2} matching $\f(u_1,\ldots,u_n)$, and $\sigma$ is the substitution $x_i\mapsto u_i$ for $i\not\in \{i_1,\ldots,i_k, j_1,\ldots,j_m\}$ and $y_{i_l}\mapsto v_{i_l}$ for $l=1,\ldots,k$.
\end{construction}
The operational model of \eqref{eq:de-simone-law-from-rules-weighted} is precisely the operational model \eqref{eq:op-model-weighted} of the given rules.

\subsection{Compositionality}
As before in the non-deterministic case (cf.\ \Cref{prop:de-simone-rule-law-affine}), to apply our theory of congruence to probabilistic De Simone specifications we need to verify naturality conditions in $\Kl(\W)$, which again boil down to the affineness of output terms required by the rule format:

\begin{proposition}\label{prop:de-simone-rule-law-affine-weighted}
For any probabilistic De Simone specification, \eqref{eq:de-simone-law-from-rules-weighted} is a De Simone law.
\end{proposition}

\noindent Our compositionality result for probabilistic trace semantics now readily emerges. Given a probabilistic De Simone specification, the operational model \eqref{eq:op-model-weighted} induces the trace map 
$\tr\colon \mS\to \W(L^{*})$.
\emph{Trace equivalence} is, as usual, the kernel relation of this map; thus two terms are trace equivalent if they generate any given completed trace with the same probability. We then obtain the following probabilistic extension of \Cref{thm:cong-de-simone-rules}:
\begin{theorem}[Compositionality of Probabilistic De Simone Rules]\label{thm:cong-de-simone-rules-weighted}
For any a.s.t.\ probabilistic De Simone specification, trace equivalence on the operational model is a congruence.
\end{theorem}

\begin{proof}
By \Cref{def:prob-de-simone}, the trace map $\tr$ is affine. By \Cref{prop:de-simone-rule-law-affine-weighted} and \Cref{thm:cong-trace-positive}, $\tr$ is a morphic congruence, whence trace equivalence is a congruence relation by \Cref{rem:congruence-relation}.
\end{proof}
A.s.t.~probabilistic De Simone specifications thus provide a congruence format for probabilistic trace semantics. To our knowledge, no such format has appeared before in the literature.

\begin{example}[Non-a.s.t.\ Specifications]\label{exa:non-ast}
Note the relevance of the a.s.t.\ assumption in \Cref{thm:cong-de-simone-rules-weighted}.
The easiest way to fail it is by involving a rule like \raisebox{1.75ex}{$\frac{}{c\xto{a/1} c}$}.
The probability of termination for the constant $c$ is $0$, and therefore the trace map is not affine. The problem 
here is that the probability of exiting the loop $c\to c$ is $0$. This can potentially be remedied 
by imposing a suitable guardedness discipline, ensuring productivity of loops, but the following example, adapted 
from~\cite[Example 20]{LevyGoncharov19}, shows that even then a.s.t.\ can fail for 
probabilistic De Simone specifications, at least for infinite signatures. Take $\Sigma X = \{c_n\mid n\in\Nat\}$, and
\begin{align*}
\inference{}{c_n\xto{a/(2^n+2)^{-1}} *}~~(n\in\Nat)\qquad \inference{}{c_n\xto{a/(1-(2^n+2)^{-1})} c_{n+1}}~~(n\in\Nat)
\end{align*}
It follows that the total weight $\tr(c_0)$ assigns to traces is $1/2$, hence 
the system is not a.s.t.
\end{example}

\takeout{
\section{Stateful De Simone Rules}
Plan: 
\begin{itemize}
\item Consider stateful process algebras (without non-affine rules like for while).
  \item One-sorted setting with functor $BX = (\Pow(\store\times (X+1)))^\store$ involving the monad $TX=(\Pow(\store\times X))^\store$. Kleisli semantics is cost semantics. (How does the corresponding stateful rule format relate to Abou-Saleh's stuff?)
\item Two-sorted reader/writer setting as in ICFP paper. Kleisli semantics should be trace semantics, but unlike in the ICFP/stateful SOS paper the general congruence result should also apply to non-deterministic languages.
\end{itemize}

OLD STUFF:
 As a more intricate showcase of the categorical machinery, we study trace semantics for imperative programming languages with effects, such as non-deterministic or probabilistic branching. 
 
 \subsection{The Language \while}
 Let us first consider a non-deterministic language. We fix a countably infinite set  $\mathcal{A}$ of \emph{(program)
    variables}, and a set $\expr$ of arithmetic expressions that are formed using
  standard operations such as $+,-,*$, constants $n\in \bZ$, and
  variables. A \emph{variable store} (or \emph{state}) is a map
  $s\colon \mathcal{A}\to\bZ$ assigning integer values to
  variables; we write $s[x\asn  n]$ for the store that maps $x$ to $n$ and otherwise equals $s$. We denote by $\store$ the set of
  states, and by ${\oname{eev}\c \expr\times \store \to \bZ}$ the \emph{expression evaluation}
  map; for example, \[\oname{eev}(x*(y+3)-y,s) = s(x)*(s(y)+3)-s(y).\]

The set $\Progs$ of program terms of \while is given by the grammar
\[    \Progs\owns p, q \Coloneqq\;\mathsf{skip}\mid x\asn e \mid \mathsf{while}\;e\;p %
      \mid p \seqcomp q \mid p+q
\]
where $x\in \mathcal{A}$ and $e\in \expr$. These operators represent, respectively, an immediately terminating program, a variable assignment, a while loop, a sequential composition, and a non-deterministic choice. The (small-step) operational semantics of \while are given by the inductive rules in \Cref{fig:rules-while}. The rules specify transitions of the form
\[ p,s\to p',s'  \qqand p,s\downarrow s'\qquad (p,p'\in\Progs,\, s,s'\in \S)\]
stating that when the program $p$ is executed on store $s$, it transforms $s$ into $s'$ and then behaves like~$p'$ (first case) or terminates (second case).
\begin{figure*}[t]
\centering
  \begin{gather*}
    \inference{p,s\to p',s'}{p \seqcomp q,s
      \to p' \seqcomp q,s'}
    \qquad
    \inference{p,s\downarrow s'}{p \seqcomp q, s
      \to q,s'}
    \qquad
    \inference{}{\mathsf{skip},s \downarrow s} \qquad 
    \inference{}{p+q,s \downarrow p,s} \qquad  \inference{}{p+q,s \downarrow q,s}
    \\[1ex]
    \inference{\oname{eev}(e,s) =
      n}{x \asn  e,s \downarrow s[x \asn  n]}%
    \qquad
    \inference{\oname{eev}(e,s) = 0}{\mathsf{while}~e~p,s \downarrow s}
    \qquad
    \inference{\oname{eev}(e,s) \not = 0}
    {\mathsf{while}~e~p,s \to
      p\seqcomp \mathsf{while}~e~p,s
    }
  \end{gather*}
  \caption{Small-step operational semantics of \while.}
  \label{fig:rules-while}
\end{figure*}
The rules determine a coalgebra
\begin{equation}\label{eq:while-opmodel} \gamma\colon \Progs \to (\Pow(\Progs\times \store+\store))^\store \end{equation}
where the set $\gamma(p)(s)$ consists of all $(p',s')\in \Progs\times \store$ with $p,s\to p',s'$ and all $s'\in \store$ with $p,s\downarrow s'$. 

We aim to study both partial and completed traces of states of \while-programs. A \emph{partial trace} of a program $p\in \Progs$ on input state $s\in \store$ is a finite sequence $s_1\cdots s_n\in S^{*}$ of states such that
\[ \exists p_1,\ldots,p_n\in \Progs.\, p,s\to p_1,s_1 \to \cdots \to p_n,s_n.\]
In particular, the empty sequence $\epsilon\in \store^{*}$ is always a partial trace. A \emph{completed trace} of $p$ on input state $s$ is a non-empty finite sequence  $s_1\cdots s_n s'\in \store^{+}$ of states such that
\[ \exists p_1,\ldots,p_n\in \Progs.\, p,s\to p_1,s_1\to\cdots \to p_n,s_n \downarrow s'.\]   
The \emph{trace map} of \while is given by
\begin{equation}\label{eq:while-trace} \tr\colon \Progs\to (\Pow(\store^{*}+\store^{+}))^\store \end{equation}
where $\tr(p)(s)$ is the set of all partial traces (in $S^{*}$) and completed traces (in $S^{+}$) at $p$ on input $s$. \emph{Trace equivalence} is the kernel of the trace map, i.e.\ the equivalence relation $\equiv$ on $\Progs$ given by
\[ p\equiv q \iff \tr(p)=\tr(q).\]
Thus two programs are trace equivalent if, for every given input state $s$, they generate the same partial traces and the same completed traces. We aim to prove the following result using our theory of compositional coalgebraic trace semantics:

\begin{theorem}[Compositionality of \while]
Trace equivalence is a congruence on the algebra $\Progs$ of \while-programs.
\end{theorem}
To apply our abstract framework, we need a coalgebraic understanding of trace semantics for \while, and we need to capture the operational rules of \while by a suitable De Simone law.

\subsection{Coalgebraic Trace Semantics for \while}\label{sec:trace-semantics-while}
Trace semantics for \while is captured by the following setting $T$, $B$, $\barB$ for coalgebraic trace semantics:
\begin{enumerate}
\item The monad $T=\Pow^\store$ on $\Set$ is the composite of the power set monad $\Pow$ with the reader monad $(-)^\store$. (This composition is known as the \emph{reader monad transformer}.) Note that \[\Kl(\Pow^\store)\cong \Kl(\Pow)^\store.\] The Kleisli category is $\DCPOb$-enriched by taking the usual enrichment of $\Kl(\Pow)$ in each component.
\item The behaviour functor $B\colon \Set\to \Set$ is given by
\[ B X = X\times \store + \store + 1.\]
The functor $B$ has the initial algebra 
 \[ \mu B = \store^{*} + \store^{+} \] 
 with structure
\[ 
\beta\colon (\store^{*}+\store^{+})\times \store + \store + 1 \xto{\cong} \store^{*}+\store^{+},\quad
 \begin{array}{lll}
(\inl(w),s) & \mapsto &\inl(sw) \\
(\inr(w),s) & \mapsto & \inr(sw) \\
s & \mapsto & \inr(s) \\ 
* & \mapsto &\inl(\epsilon).
\end{array}         
\]

\item The Kleisli extension $\barB\colon \Kl(\Pow^\store)\to \Kl(\Pow^\store)$ is determined by the distributive law
$\delta^{B}\colon B\Pow^\store\to \Pow^\store B$ whose components 
\[
\delta^{B}\colon (\Pow X)^\store\times \store + \store + 1 \to (\Pow(X\times\store + \store + 1))^\store  
\]
are given by 
\begin{align*}
& (f,s) && \mapsto & \lambda t.\, \{ (u,s) \mid u\in f(t) \} \\
& s && \mapsto & \lambda t.\ \{s\} \\
& * && \mapsto & \{ * \}.  
\end{align*}
\end{enumerate}
Then the map \eqref{eq:while-opmodel} form a $TB$-coalgebra
\begin{equation}\label{eq:while-opmodel-coalg}
\gamma\colon \Progs \to (\Pow B\Progs)^\store = TB\Progs,
\end{equation} 
and its associated coalgebraic trace morphism
\[ \tr\colon \Progs\to (\Pow(\store^{*}+\store^{+}))^\store=T(\mu B)  \]
is the map \eqref{eq:whiletwo-trace}. Indeed, one readily verifies that this map forms a $\barB$-coalgebra morphism from the coalgebra $(\Progs,\gamma)$ of programs to the final coalgebra $(S^{*}+S^{+}, J\beta^{-1})$ of traces.

\subsection{De Simone Law for \whiletwo}

To capture \whiletwo in the framework of De Simone laws, we instantiate the latter to be monad $T$ and behaviour functor $B$ as in \Cref{sec:trace-semantics-whiletwo}, and the polynomial functor $\Sigma\colon \Set^2\to \Set^2$ corresponding to the two-sorted signature of $\whiletwo$, with sorts $\Trs$ and $\Cos$:
\begin{align*}
  \Sigma_{\Trs} X
  &= \underbrace{\term}_{\mathsf{skip}}
  +\, \underbrace{\mathcal{A} \times \expr}_{\asn }
  \,+\, \underbrace{\expr \times X_{\Trs}}_{\mathsf{while}}
  \,+\, \underbrace{X_{\Trs} \times X_{\Trs}}_{-\seqcomp-} \,+\, \underbrace{X_{\Trs} \times X_{\Trs}}_{-+-} \\
  \Sigma_{\Cos} X
 & = \underbrace{X_{\Trs}\times\store}_{[-]_s}
  \,+\, 
  \underbrace{X_{\Cos} \times X_{\Trs}}_{-\seqcomp-} \,+\, \underbrace{X_{\Cos} \times X_{\Trs}}_{-+-}.
\end{align*}
Thus the initial algebra $\mS$ is the two-sorted algebra $(\Tr,\Co)$ of \whiletwo-programs. On top of the above distributive law $\delta^{B}$), we take the distributive law $\delta^{\Sigma}\colon \Sigma T\to T\Sigma$ which is given by the maps
 \[(\delta^{\Sigma})^\Trs\colon \Sigma_\Trs TX \to \Pow \Sigma_\Trs X \qqand  (\delta^{\Sigma})^\Cos\colon \Sigma_\Cos TX \to \Pow \Sigma_\Cos X\]
 defined for $\f$ in $\Sigma_\Trs$ or $\Sigma_\Cos$ (resp.) by
 \[ \f(U_1,\ldots,U_n) \mapsto \{ \f(u_1,\ldots,u_n) \mid u_i\in U_i\text{ for all $i$} \}.\]

The operational rules of \whiletwo (\Cref{fig:comp-rules}) can be presented as the De Simone law of $\Sigma$ over $B$, that is, a natural transformation
\[ \rho\colon \Sigma(X+BX)\to TB\Sigmas X \qquad (X\in \Set^2),  \]
that encodes the rules into functions. Formally, in sort $\Trs$, the map
\[ \rho_X^\Trs\colon \Sigma_\Trs (X+BX) \to \Pow((\Sigmas_\Cos X \times \store+\store)^\store)  \]
is given as follows for $p,q\in X_\Trs$, $f\in (X_\Cos\times \store+\store)^\store$ (we identify elements of the set $(\Sigmas_\Cos X \times \store+\store)^\store$ with singleton subsets):
\begin{align*}
&f \seqcomp q &&\mapsto&& \lambda s.\, \begin{cases}
(c\seqcomp q,s') & \text{if $f(s)=(c,s')$} \\
([q]_{s'},s') & \text{if $f(s)=s'$}
\end{cases}\\
&p+q && \mapsto&& \lambda s.\, ([p]_s+[q]_s,s) \\
&\mathsf{skip} && \mapsto&& \lambda s.\, s \\
&x\asn e && \mapsto&& \lambda s.\, s[x\asn n] \quad \text{where $\mathsf{eev}(e,s)=n$}\\
&\mathsf{while}~e~p && \mapsto && \lambda s. \begin{cases}
s & \text{if $\mathsf{eev}(e,s)=0$} \\
\run{s}{p\seqcomp \mathsf{while}~e~p} & \text{if $\mathsf{eev}(e,s)\neq 0$}
\end{cases} \\
&\_ &&\mapsto && \emptyset
\end{align*}
The last item means that every input term that is not matched by any other case (e.g.~$p\seqcomp q$ or $p\seqcomp f$) is mapped to the empty set.

Similarly, in sort $\Cos$, the map 
\[ \rho_X^\Cos\colon \Sigma_\Cos (X+BX) \to \Pow(\Sigmas_\Cos X \times \store+\store) \]
is given for $f\in (X_\Cos\times \store+\store)^\store$, $q\in X_\Trs$, $U,V\in \Pow(X_\Cos\times \store + \store)$ by
\begin{align*}
&\run{s}{f}  && \mapsto&& \{ f(s) \} \\
& U \seqcomp q && \mapsto &&  \{ (d\seqcomp q,s) \mid (d,s)\in U \} \cup \{ \run{s'}{q}\mid s'\in U\}  \\
& U + d && \mapsto && U  \\
& c + V && \mapsto && V  \\
& \_ && \mapsto && \emptyset
\end{align*}
Since the above De Simone law $\rho$ matches precisely the operational rules of \whiletwo, its operational model is the $TB$-coalgebra \eqref{eq:whiletwo-opmodel} obtained by running programs according to those rules. 

In order to apply our theory of congruence, we need to establish an affinity property for $\rho$. This is slightly more subtle than in the case of standard De Simone rules (\Cref{prop:de-simone-rule-law-affine}): while all the writer rules in \Cref{fig:comp-rules} only produce affine terms, the second rule for the reader $\mathsf{while}~e~p$ in \Cref{fig:rules-while} produces the non-affine term $\run{s}{p \seqcomp \mathsf{while}~e~p}$. This is remedied by the fact that readers are deterministic, yielding the following categorical affinity statement:

\begin{proposition}\label{prop:rho-whiletwo-affine}
The De Simone Law $\rho$ for $\whiletwo$ is affine w.r.t.\ all morphisms $p\colon P\kleislito X$ in $\Kl(T)$ that are deterministic in the reader component and right-total in the writer component.
\end{proposition}
The conditions on $p$ mean that $p^\Trs(z)\in \Pow_1 X_\Trs$ for all $z\in P_\Trs$ and $p^\Cos(z)\in \Pow_{+}X_\Cos$ for all $z\in P_\Cos$. Here $\Pow_1 Y$ and $\Pow_{+} Y$ denote the sets of singleton and non-empty subsets of a set $Y$, respectively.

\begin{proof}
We consider the reader and writer components of $\rho$ separately. For the former, $\f\colon \Trs\times\cdots\times \Trs\to \Trs$ be a reader operation in $\Sigma$. (Note that all operations of output type $\Trs$ also have inputs of type~$\Trs$.) The affineness condition then amounts to the statement that the two legs of the diagram below send every element $\f(U_1,\ldots,U_n)$ of $\Sigma_\Trs(\Pow^2 X+ B\Pow^2 X)$, where $U_i\in \Pow_1 X_\Trs + (\Pow_{+} X_\Cos \times\store+\store)^\store$, to the same element of $\Pow ((\Sigmas_\Cos\times \store+\store)^\store)$.   
\[
\begin{tikzcd}
\Sigma_\Trs(\Pow^2 X+ B\Pow^2 X) \ar{d}[swap]{\Sigma_\Trs\delta^{B_0}} \ar{r}{\rho_{\Pow^2 X}^\Trs} & \Pow( (\Sigmas_\Cos \Pow^2 X \times \store+\store)^\store) \ar{d}{\Pow B_\Trs\delta^{\Sigmas}}  \\
\Sigma_\Trs \Pow^2 (X+BX) \ar{d}[swap]{\delta^\Trs_{\Sigma,\Pow^2 }} & \Pow ((\Pow\Sigmas_\Cos X \times \store+\store)^\store) \ar{d}{\Pow \delta^\Trs_{B,\Pow^2 }} \\
\Pow \Sigma_\Trs (X+BX) \ar{r}{\Pow \rho_X^\Trs} & \Pow \Pow ((\Sigmas_\Cos\times \store+\store)^\store) \ar{d}{\mu} \\
& \Pow ((\Sigmas_\Cos\times \store+\store)^\store) 
\end{tikzcd}
\]
As indicated above, the critical case is that of a term $\mathsf{while}~e~U_1$. We consider two cases:
\begin{enumerate}
\item If $U_1\in \Pow_1 X_\Trs$, say $U_1=\{p\}$, then the term is sent by both legs to the singleton $\{F\}$ where
\[ F(s) = \begin{cases}
(p,s) & \text{if $\mathsf{eev}(e,s)\neq 0$}\\
s & \text{if $\mathsf{eev}(e,s)=0$}
\end{cases} \]
\item If $U_1\in (\Pow_{+} X_\Cos \times\store+\store)^\store$, then the term is sent by both legs to~$\emptyset$, corresponding to the ``\_'' case in the definition of $\rho^\Trs$.
\end{enumerate}
For all other reader terms, the output terms of the corresponding rules are affine and thus we can the reasoning exactly as in \Cref{prop:de-simone-rule-law-affine}. The same reasoning applies to writer terms.
\end{proof}

}

\section{Conclusion and Future Work}
We have developed a theory of compositional trace semantics for process specifications, emerging at the generality of \emph{abstract GSOS}~\cite{DBLP:conf/lics/TuriP97} along with \emph{coalgebraic trace semantics}~\cite{hjs07}. One important message of our paper is that `vanilla' abstract GSOS is more powerful than it seems: unlike existing categorical approaches to the subject (see Related Work), our results are based on a slightly refined version of Turi and Plotkin's original results, applied to our novel notion of De Simone laws over Kleisli categories. We have illustrated the usefulness of this abstract framework by deriving from it a novel De Simone-type format for probabilistic LTS. Various subtleties of the format (e.g.~regarding almost sure termination) would be hard to uncover from scratch, but are naturally explained and motivated by the categorical setup.

Within the probabilistic setting, our compositionality result is currently limited to almost-surely terminating systems. 
In future work, we are planning to analyse this issue through the lens of (coalgebraic) \emph{infinite trace semantics}~\cite{DBLP:journals/entcs/Jacobs04a}, which remedies ``probability leakage'', caused by non-termination.
Kock's synthetic approach to measurability and probability~\cite{Kock2012} suggests that our framework may be applicable to more sophisticated probabilistic settings, via suitable additive monads. 
An orthogonal direction is to reformulate our framework internally to a \emph{Markov category}~\cite{ChoJacobs2019,Fritz2020}, instead of a Kleisli category of a monad.

We also aim to apply our framework to computational effects beyond probabilities and cover other features, such as states or higher-order behaviour.
These have been addressed recently from the perspective of abstract GSOS~\cite{DBLP:conf/fscd/0001MS0U22,gmstu25,gmstu23} under (weak) bisimilarity, but a corresponding abstract treatment of trace semantics is still missing. Another direction are congruence formats for notions of \emph{trace distance} in lieu of \emph{trace equivalence}. This will require departing from final coalgebra semantics in favour of a fibrational approach~\cite{DBLP:journals/lmcs/BaldanBKK18,DBLP:conf/concur/Bonchi0P18, 10.1145/3776697}.

\clearpage
\bibliographystyle{plainurl}%

\bibliography{mainBiblio}

@PREAMBLE{ {\providecommand{\noopsort}[1]{}} }

@string{acm="ACM"}

@string{lncs="LNCS"}

@string{entcs="ENTCS"}

@string{elsevier="Elsevier"}

@string{springer="Springer"}

@string{lipics="LIPIcs"}

@string{dagstuhl="Schloss Dagstuhl -- Leibniz-Zentrum f{\"u}r Informatik"}

@article{DBLP:journals/lmcs/BaldanBKK18,
  author       = {Paolo Baldan and
                  Filippo Bonchi and
                  Henning Kerstan and
                  Barbara K{\"{o}}nig},
  title        = {Coalgebraic Behavioral Metrics},
  journal      = {Log. Methods Comput. Sci.},
  volume       = {14},
  number       = {3},
  year         = {2018},
}

@inproceedings{DBLP:conf/concur/Bonchi0P18,
  author       = {Filippo Bonchi and
                  Barbara K{\"{o}}nig and
                  Daniela Petrisan},
  editor       = {Sven Schewe and
                  Lijun Zhang},
  title        = {Up-To Techniques for Behavioural Metrics via Fibrations},
  booktitle    = {29th International Conference on Concurrency Theory, {CONCUR} 2018,
                  Beijing, China, September 4-7, 2018},
  series       = {LIPIcs},
  pages        = {17:1--17:17},
  publisher    = {Schloss Dagstuhl - Leibniz-Zentrum f{\"{u}}r Informatik},
  year         = {2018},
}

@article{Klin09,
  author       = {Bartek Klin},
  title        = {Bialgebraic methods and modal logic in structural operational semantics},
  journal      = {Inf. Comput.},
  volume       = {207},
  number       = {2},
  pages        = {237--257},
  year         = {2009},
}

@inproceedings{Klin10,
  author       = {Bartek Klin},
  OPTeditor       = {Bart Jacobs and
                  Milad Niqui and
                  Jan J. M. M. Rutten and
                  Alexandra Silva},
  title        = {Structural Operational Semantics and Modal Logic, Revisited},
  booktitle    = {Proceedings of the Tenth Workshop on Coalgebraic Methods in Computer
                  Science, CMCS@ETAPS 2010, Paphos, Cyprus, March 26-28, 2010},
  series       = {Electronic Notes in Theoretical Computer Science},
  number       = {2},
  pages        = {155--175},
  publisher    = {Elsevier},
  year         = {2010},
}

@InProceedings{glabbeek90,
author="van Glabbeek, R. J.",
OPTeditor="Baeten, J. C. M.
and Klop, J. W.",
title="The linear time - branching time spectrum",
booktitle="CONCUR '90 Theories of Concurrency: Unification and Extension",
year="1990",
publisher=springer,
pages="278--297",
}

@incollection{cj13,
  author       = {Dion Coumans and
                  Bart Jacobs},
  OPTeditor       = {Chris Heunen and
                  Mehrnoosh Sadrzadeh and
                  Edward Grefenstette},
  title        = {Scalars, Monads, and Categories},
  booktitle    = {Quantum Physics and Linguistics - {A} Compositional, Diagrammatic
                  Discourse},
  pages        = {184--216},
  publisher    = {Oxford University Press},
  year         = {2013},
  doi          = {10.1093/ACPROF:OSO/9780199646296.003.0007},
}

@ARTICLE{Lindner1979-qe,
  title     = "Affine parts of monads",
  author    = "Lindner, Harald",
  journal   = "Arch. Math.",
  publisher = "Springer Nature",
  volume    =  33,
  number    =  1,
  pages     = "437--443",
  month     =  dec,
  year      =  1979,
  language  = "en"
}

@article{bloom94,
author = {Bloom, Bard},
title = {When is partial trace equivalence adequate?},
year = {1994},
publisher = {Springer-Verlag},
address = {Berlin, Heidelberg},
volume = {6},
number = {3},
doi = {10.1007/BF01215409},
pages = {317–338},
}

@article{DeSimone1985,
  author    = {R. de Simone},
  title     = {Higher-Level Synchronising Devices in Meije-SCCS},
  journal   = {Theoretical Computer Science},
  year      = {1985},
  volume    = {37},
  number    = {3},
  pages     = {245--267}
}

@inproceedings{mulry93,
  author       = {Philip S. Mulry},
  OPTeditor       = {Stephen D. Brookes and
                  Michael G. Main and
                  Austin Melton and
                  Michael W. Mislove and
                  David A. Schmidt},
  title        = {Lifting Theorems for Kleisli Categories},
  booktitle    = {Mathematical Foundations of Programming Semantics, 9th International
                  Conference, New Orleans, LA, USA, April 7-10, 1993, Proceedings},
  series       = {Lecture Notes in Computer Science},
  volume       = {802},
  pages        = {304--319},
  publisher    = {Springer},
  year         = {1993},
  url          = {https://doi.org/10.1007/3-540-58027-1\_15},
  doi          = {10.1007/3-540-58027-1\_15},
  bibsource    = {dblp computer science bibliography, https://dblp.org}
}

@inproceedings{fmu22,
  author       = {Florian Frank and
                  Stefan Milius and
                  Henning Urbat},
  OPTeditor       = {Helle Hvid Hansen and
                  Fabio Zanasi},
  title        = {Coalgebraic Semantics for Nominal Automata},
  booktitle    = {Coalgebraic Methods in Computer Science - 16th {IFIP} {WG} 1.3 International
                  Workshop, {CMCS} 2022, Colocated with {ETAPS} 2022, Munich, Germany,
                  April 2-3, 2022, Proceedings},
  series       = {Lecture Notes in Computer Science},
  volume       = {13225},
  pages        = {45--66},
  publisher    = {Springer},
  year         = {2022},
  url          = {https://doi.org/10.1007/978-3-031-10736-8\_3},
  doi          = {10.1007/978-3-031-10736-8\_3},
  bibsource    = {dblp computer science bibliography, https://dblp.org}
}

@misc{gmstu25,
	title={Bialgebraic Reasoning on Stateful Languages}, 
	author={Sergey Goncharov and Stefan Milius and Lutz Schröder and Stelios Tsampas and Henning Urbat},
	year={2025},
	eprint={2503.10955},
	archivePrefix={arXiv},
	primaryClass={cs.PL},
  note = {To appear in Proc.~ICFP 2025}
}

@InProceedings{amm21,
  author =	{Ad\'{a}mek, Ji\v{r}{\'\i} and Milius, Stefan and Moss, Lawrence S.},
  title =	{{Initial Algebras Without Iteration}},
  booktitle =	{9th Conference on Algebra and Coalgebra in Computer Science (CALCO 2021)},
  pages =	{5:1--5:20},
  series =	{Leibniz International Proceedings in Informatics (LIPIcs)},
  year =	{2021},
  volume =	{211},
  editor =	{Gadducci, Fabio and Silva, Alexandra},
  publisher =	{Schloss Dagstuhl -- Leibniz-Zentrum f{\"u}r Informatik},
}

@inproceedings{gmstu23,
author = {Goncharov, Sergey and Milius, Stefan and Schr\"{o}der, Lutz and Tsampas, Stelios and Urbat, Henning},
title = {Towards a Higher-Order Mathematical Operational Semantics},
booktitle = {50th ACM SIGPLAN Symposium on Principles of Programming Languages (POPL 2023)},
year = {2023},
issue_date = {January 2023},
publisher = acm,
OPTaddress = {New York, NY, USA},
volume = {7},
OPTnumber = {POPL},
doi = {10.1145/3571215},
series = {Proc. ACM Program. Lang.},
OPTmonth = {Jan},
articleno = {22},
numpages = {27}
}

@InProceedings{   LevyGoncharov19,
  author        = {Paul Blain Levy and Sergey Goncharov},
  title         = {Coinductive Resumption Monads: Guarded Iterative and
                  Guarded Elgot},
  booktitle     = {Proc.~8rd international conference on Algebra and
                  coalgebra in computer science (CALCO 2019)},
  year          = {2019},
  nopages       = {34--48},
  series        = {LIPIcs},
  novolume      = {35},
  publisher     = {Schloss Dagstuhl - Leibniz-Zentrum für Informatik}
}

@Article{Barr70,
  author = 	 {Michael Barr},
  title = 	 {Coequalizers and free triples},
  journal = 	 {Math.~Z.},
  year = 	 {1970},
  OPTkey = 	 {},
  volume = 	 {116},
  OPTnumber = 	 {},
  pages = 	 {307-322},
  OPTmonth = 	 {},
  OPTnote = 	 {},
  OPTannote = 	 {}
}

@Article{Lambek68,
  author = 	 {Lambek, Joachim},
  title = 	 {A Fixpoint Theorem for Complete Categories},
  journal = 	 {Math.~Z.},
  year = 	 {1968},
  OPTkey = 	 {},
  volume = 	 {103},
  OPTnumber = 	 {},
  pages = 	 {151--161},
  OPTmonth = 	 {},
  OPTnote = 	 {},
  OPTannote = 	 {}
}

@article{DBLP:journals/tcs/Rutten00,
  author    = {Jan J. M. M. Rutten},
  title     = {Universal coalgebra: a theory of systems},
  journal   = {Theor. Comput. Sci.},
  volume    = {249},
  number    = {1},
  pages     = {3--80},
  year      = {2000},
  url       = {https://doi.org/10.1016/S0304-3975(00)00056-6},
  doi       = {10.1016/S0304-3975(00)00056-6},
  bibsource = {dblp computer science bibliography, https://dblp.org}
}

@article{DBLP:journals/jacm/BloomIM95,
  author    = {Bard Bloom and
               Sorin Istrail and
               Albert R. Meyer},
  title     = {Bisimulation Can't be Traced},
  journal   = {J. {ACM}},
  volume    = {42},
  number    = {1},
  pages     = {232--268},
  year      = {1995},
  url       = {https://doi.org/10.1145/200836.200876},
  doi       = {10.1145/200836.200876},
  bibsource = {dblp computer science bibliography, https://dblp.org}
}

@article{DBLP:journals/tcs/Klin11,
  author    = {Bartek Klin},
  title     = {Bialgebras for structural operational semantics: An introduction},
  journal   = {Theor. Comput. Sci.},
  volume    = {412},
  number    = {38},
  pages     = {5043--5069},
  year      = {2011},
  url       = {https://doi.org/10.1016/j.tcs.2011.03.023},
  doi       = {10.1016/j.tcs.2011.03.023},
  bibsource = {dblp computer science bibliography, https://dblp.org}
}

@book{DBLP:books/cu/J2016,
  author    = {Bart Jacobs},
  title     = {{Introduction to Coalgebra: Towards Mathematics of States and Observation}},
  series    = {Cambridge Tracts in Theoretical Computer Science},
  volume    = {59},
  publisher = {Cambridge University Press},
  year      = {2016},
  doi       = {10.1017/CBO9781316823187},
  isbn      = {9781316823187},
  bibsource = {dblp computer science bibliography, https://dblp.org}
}

@article{CASTIGLIONI2024100929,
title = {Back to the format: A survey on SOS for probabilistic processes},
journal = {Journal of Logical and Algebraic Methods in Programming},
volume = {137},
pages = {100929},
year = {2024},
issn = {2352-2208},
doi = {https://doi.org/10.1016/j.jlamp.2023.100929},
author = {Valentina Castiglioni and Ruggero Lanotte and Simone Tini},
}

@article{DBLP:journals/iandc/GrooteV92,
  author    = {Jan Friso Groote and
               Frits W. Vaandrager},
  title     = {Structured Operational Semantics and Bisimulation as a Congruence},
  journal   = {Inf. Comput.},
  volume    = {100},
  number    = {2},
  pages     = {202--260},
  year      = {1992},
  url       = {https://doi.org/10.1016/0890-5401(92)90013-6},
  doi       = {10.1016/0890-5401(92)90013-6},
  bibsource = {dblp computer science bibliography, https://dblp.org}
}

@inproceedings{DBLP:conf/lics/TuriP97,
  author    = {Daniele Turi and
               Gordon D. Plotkin},
  title     = {Towards a Mathematical Operational Semantics},
  booktitle = {12th Annual {IEEE} Symposium on Logic in Computer Science (LICS 1997)},
  pages     = {280--291},
  year      = {1997},
  crossref-ignore  = {DBLP:conf/lics/1997},
  doi       = {10.1109/LICS.1997.614955},
  bibsource = {dblp computer science bibliography, https://dblp.org}
}

@inproceedings{DBLP:conf/ctcs/Turi97,
  author    = {Daniele Turi},
  title     = {Categorical Modelling of Structural Operational Rules: Case Studies},
  booktitle = {Category Theory and Computer Science, 7th International Conference,
               {CTCS} '97, Santa Margherita Ligure, Italy, September 4-6, 1997, Proceedings},
  pages     = {127--146},
  year      = {1997},
  crossref-ignore  = {DBLP:conf/ctcs/1997},
  url       = {https://doi.org/10.1007/BFb0026985},
  doi       = {10.1007/BFb0026985},
  bibsource = {dblp computer science bibliography, https://dblp.org}
}

@article{hjs07,
  author    = {Ichiro Hasuo and
               Bart Jacobs and
               Ana Sokolova},
  title     = {Generic Trace Semantics via Coinduction},
  journal   = {Logical Methods in Computer Science},
  volume    = {3},
  number    = {4},
  year      = {2007},
  url       = {https://doi.org/10.2168/LMCS-3(4:11)2007},
  doi       = {10.2168/LMCS-3(4:11)2007},
  bibsource = {dblp computer science bibliography, https://dblp.org}
}

@inproceedings{DBLP:journals/entcs/Abou-SalehP11,
  author    = {Faris Abou{-}Saleh and
               Dirk Pattinson},
  editor    = {Michael W. Mislove and
               Jo{\"{e}}l Ouaknine},
  title     = {Towards Effects in Mathematical Operational Semantics},
  booktitle = {Mathematical Foundations of Programming
               Semantics, {MFPS} 2011},
  series    = entcs,
  volume    = {276},
  pages     = {81--104},
  publisher = elsevier,
  year      = {2011},
  nourl       = {https://doi.org/10.1016/j.entcs.2011.09.016},
  doi       = {10.1016/j.entcs.2011.09.016},
  bibsource = {dblp computer science bibliography, https://dblp.org}
}

@phdthesis{DBLP:phd/ethos/AbouSaleh14,
  author    = {Faris Abou{-}Saleh},
  title     = {A coalgebraic semantics for imperative programming languages},
  school    = {Imperial College London, {UK}},
  year      = {2014},
  url       = {http://hdl.handle.net/10044/1/13693},
  bibsource = {dblp computer science bibliography, https://dblp.org}
}

@article{DBLP:journals/entcs/Jacobs04a,
  author    = {Bart Jacobs},
  title     = {Trace Semantics for Coalgebras},
  journal   = {Electr. Notes Theor. Comput. Sci.},
  volume    = {106},
  pages     = {167--184},
  year      = {2004},
  url       = {https://doi.org/10.1016/j.entcs.2004.02.031},
  doi       = {10.1016/j.entcs.2004.02.031},
  bibsource = {dblp computer science bibliography, https://dblp.org}
}

@article{DBLP:journals/jlp/BrengosMP15,
  author    = {Tomasz Brengos and
               Marino Miculan and
               Marco Peressotti},
  title     = {Behavioural equivalences for coalgebras with unobservable moves},
  journal   = {J. Log. Algebraic Methods Program.},
  volume    = {84},
  number    = {6},
  pages     = {826--852},
  year      = {2015},
  url       = {https://doi.org/10.1016/j.jlamp.2015.09.002},
  doi       = {10.1016/j.jlamp.2015.09.002},
  bibsource = {dblp computer science bibliography, https://dblp.org}
}

@phdthesis{56f40c248cb44359beb3c28c3263838e,
title = "On generalised coinduction and probabilistic specification formats: Distributive laws in coalgebraic modelling",
author = "Falk Bartels",
year = "2004",
language = "English",
series = "IPA dissertation series",
number = "6",
school = "Vrije Universiteit Amsterdam",
}

@inproceedings{DBLP:conf/mfcs/0001WNDP21,
  author    = {Stelios Tsampas and
               Christian Williams and
               Andreas Nuyts and
               Dominique Devriese and
               Frank Piessens},
  editor    = {Filippo Bonchi and
               Simon J. Puglisi},
  title     = {Abstract Congruence Criteria for Weak Bisimilarity},
  booktitle = {46th International Symposium on Mathematical Foundations of Computer Science, MFCS'21},
  series    = {LIPIcs},
  volume    = {202},
  pages     = {88:1--88:23},
  publisher = {Schloss Dagstuhl - Leibniz-Zentrum f{\"{u}}r Informatik},
  year      = {2021},
  url       = {https://doi.org/10.4230/LIPIcs.MFCS.2021.88},
  doi       = {10.4230/LIPIcs.MFCS.2021.88},
  bibsource = {dblp computer science bibliography, https://dblp.org}
}

@book{mac2013categories,
  title = {Categories for the {{Working Mathematician}}},
  author = {Mac Lane, S.},
  year = {1978},
  series = {Graduate {{Texts}} in {{Mathematics}}},
  edition = {2},
  volume = {5},
  publisher = springer,
  OPTaddress = {{New York}},
  url = {http://link.springer.com/10.1007/978-1-4757-4721-8},
  isbn = {978-0-387-98403-2},
  langid = {english}
}

@book{dkv09,
  editor    = {Droste, Manfred and Kuich, Werner and Vogler, Heiko},
  title     = {Handbook of Weighted Automata},
  publisher = {Springer},
  year      = {2009},
}

@inproceedings{bartels02,
  author       = {Falk Bartels},
  editor       = {Lawrence S. Moss},
  title        = {{GSOS} for Probabilistic Transition Systems},
  booktitle    = {Coalgebraic Methods in Computer Science, {CMCS} 2002, Satellite Event
                  of {ETAPS} 2002, Grenoble, France, April 6-7, 2002},
  series       = {Electronic Notes in Theoretical Computer Science},
  volume       = {65},
  number       = {1},
  pages        = {29--53},
  publisher    = {Elsevier},
  year         = {2002},
  doi          = {10.1016/S1571-0661(04)80358-X},
}

@article{ks13,
title = {Structural operational semantics for stochastic and weighted transition systems},
journal = {Information and Computation},
volume = {227},
pages = {58-83},
year = {2013},
doi = {https://doi.org/10.1016/j.ic.2013.04.001},
}

@inproceedings{DBLP:conf/fossacs/KlinS08,
  author    = {Bartek Klin and
               Vladimiro Sassone},
  editor    = {Roberto M. Amadio},
  title     = {Structural Operational Semantics for Stochastic Process Calculi},
  booktitle = {11th International Conference Foundations of Software Science and Computational Structures, FOSSACS'08},
  series    = lncs,
  volume    = {4962},
  pages     = {428--442},
  publisher = {Springer},
  year      = {2008},
  url       = {https://doi.org/10.1007/978-3-540-78499-9\_30},
  doi       = {10.1007/978-3-540-78499-9\_30},
  bibsource = {dblp computer science bibliography, https://dblp.org}
}

@article{DBLP:journals/tcs/MiculanP16,
  author    = {Marino Miculan and
               Marco Peressotti},
  title     = {Structural operational semantics for non-deterministic processes with
               quantitative aspects},
  journal   = {Theor. Comput. Sci.},
  volume    = {655},
  pages     = {135--154},
  year      = {2016},
  url       = {https://doi.org/10.1016/j.tcs.2016.01.012},
  doi       = {10.1016/j.tcs.2016.01.012},
  bibsource = {dblp computer science bibliography, https://dblp.org}
}

@inproceedings{DBLP:conf/fossacs/Abou-SalehP13,
  author    = {Faris Abou{-}Saleh and
               Dirk Pattinson},
  editor    = {Frank Pfenning},
  title     = {Comodels and Effects in Mathematical Operational Semantics},
  booktitle = {Foundations of Software Science and Computation Structures - 16th
               International Conference, {FOSSACS} 2013, Held as Part of the European
               Joint Conferences on Theory and Practice of Software, {ETAPS} 2013,
               Rome, Italy, March 16-24, 2013. Proceedings},
  series    = {Lecture Notes in Computer Science},
  volume    = {7794},
  pages     = {129--144},
  publisher = {Springer},
  year      = {2013},
  url       = {https://doi.org/10.1007/978-3-642-37075-5\_9},
  doi       = {10.1007/978-3-642-37075-5\_9},
  bibsource = {dblp computer science bibliography, https://dblp.org}
}

@inproceedings{DBLP:conf/fscd/0001MS0U22,
  author    = {Sergey Goncharov and
               Stefan Milius and
               Lutz Schr{\"{o}}der and
               Stelios Tsampas and
               Henning Urbat},
  editor    = {Amy P. Felty},
  title     = {Stateful Structural Operational Semantics},
  booktitle = {7th International Conference on Formal Structures for Computation and Deduction, FSCD'22},
  series    = {LIPIcs},
  volume    = {228},
  pages     = {30:1--30:19},
  publisher = {Schloss Dagstuhl - Leibniz-Zentrum f{\"{u}}r Informatik},
  year      = {2022},
  url       = {https://doi.org/10.4230/LIPIcs.FSCD.2022.30},
  doi       = {10.4230/LIPIcs.FSCD.2022.30},
  bibsource = {dblp computer science bibliography, https://dblp.org}
}

@incollection{DBLP:books/el/01/AcetoFV01,
  author    = {Luca Aceto and
               Wan J. Fokkink and
               Chris Verhoef},
  editor    = {Jan A. Bergstra and
               Alban Ponse and
               Scott A. Smolka},
  title     = {Structural Operational Semantics},
  booktitle = {Handbook of Process Algebra},
  pages     = {197--292},
  publisher = {North-Holland / Elsevier},
  year      = {2001},
  url       = {https://doi.org/10.1016/b978-044482830-9/50021-7},
  doi       = {10.1016/b978-044482830-9/50021-7},
  bibsource = {dblp computer science bibliography, https://dblp.org}
}

@InProceedings{10.1007/978-3-642-32784-1_5,
author="Bacci, Giorgio
and Miculan, Marino",
editor="Pattinson, Dirk
and Schr{\"o}der, Lutz",
title="Structural Operational Semantics for Continuous State Probabilistic Processes",
booktitle="Coalgebraic Methods in Computer Science",
year="2012",
publisher="Springer Berlin Heidelberg",
address="Berlin, Heidelberg",
pages="71--89",
isbn="978-3-642-32784-1"
}

@article{10.1145/3776697,
author = {Urbat, Henning},
title = {Higher-Order Behavioural Conformances via Fibrations},
year = {2026},
issue_date = {January 2026},
publisher = {Association for Computing Machinery},
address = {New York, NY, USA},
volume = {10},
number = {POPL},
url = {https://doi.org/10.1145/3776697},
doi = {10.1145/3776697},
journal = {Proc. ACM Program. Lang.},
month = jan,
articleno = {55},
numpages = {29}
}

@article{Kock2012,
  author    = {Anders Kock},
  title     = {Commutative monads as a theory of distributions},
  journal   = {Theory and Applications of Categories},
  volume    = {26},
  number    = {4},
  year      = {2012},
  pages     = {97--131},
  url       = {http://www.tac.mta.ca/tac/volumes/26/4/26-04abs.html}
}

@article{ChoJacobs2019,
  author    = {Kenta Cho and Bart Jacobs},
  title     = {Disintegration and {B}ayesian Inversion via String Diagrams},
  journal   = {Mathematical Structures in Computer Science},
  volume    = {29},
  year      = {2019},
  doi       = {10.1017/S0960129518000488},
  note      = {arXiv:1709.00322}
}

@article{Fritz2020,
  author    = {Tobias Fritz},
  title     = {A synthetic approach to {M}arkov kernels, conditional independence and theorems on sufficient statistics},
  journal   = {Advances in Mathematics},
  volume    = {370},
  year      = {2020},
  doi       = {10.1016/j.aim.2020.107239},
  note      = {arXiv:1908.07021}
}

\clearpage
\appendix
\onecolumn

\section{Appendix}
This appendix contains omitted proofs and additional details.

\begin{rem}[Recursion]\label{rem:primitive-recursion}
Let $\Sigma\colon \C\to \C$ be an endofunctor with a free algebra $\Sigmas X$ generated by $X\in \C$. Its universal property entails a useful recursion principle (cf.~\cite[Prop.~2.4.6]{DBLP:books/cu/J2016}). For every $\Sigma$-algebra morphism $f\colon (\Sigmas X,\ini_X)\to (A,a)$, every morphism $g\colon X\to C$ and every morphism $c\colon \Sigma(A\times C)\to C$, there exists a unique morphism $h$ making the diagram below commute:
\[
\begin{tikzcd}[column sep=40]
X \ar{r}{\eta} \ar{dr}[swap]{g} & \Sigmas X \ar[dashed]{d}{h} & \Sigma\Sigmas X \ar{l}[swap]{\ini_X}  \ar{d}{\Sigma\langle f,\,h\rangle} \\
& C & \Sigma(A\times C) \ar{l}[swap]{c}
\end{tikzcd}
\]
\end{rem}

\subsection[Proof of Theorem~\ref{thm:cong}]{Proof of \Cref{thm:cong}}

The key construction for our proof is an extension of pre-GSOS laws to the free monad~$\Sigmas$:
\begin{notation}
Every pre-GSOS law \eqref{eq:pre-gsos-law} induces a family of morphisms
\begin{equation}\label{eq:pre-gsos-law-ext} \rho_X^{*}\colon \Sigmas(X\times BX)\to B\Sigmas X \qquad (X\in \C) \end{equation}
defined recursively (\Cref{rem:primitive-recursion}) as the unique morphism making the diagram below commute:
\[
\begin{tikzcd}[column sep=10]
X\times BX \ar{r}{\eta} \ar{d}[swap]{\outr} & \Sigmas(X\times BX) \ar[dashed]{d}{\rho^{*} } && \Sigma\Sigmas(X\times BX) \ar{ll}[swap]{\ini}  \ar{d}{\Sigma \langle \Sigmas \outl, \, \rho^{*}\rangle} \\
BX \ar{r}{B\eta_X} & B\Sigmas X & B\Sigmas\Sigmas X \ar{l}[swap]{B\mu} & \Sigma(\Sigmas X\times B\Sigmas X)  \ar{l}[swap]{\rho}
\end{tikzcd}
\]
\end{notation}
Informally, while the original law $\rho$ specifies the behaviour of individual operations from $\Sigma$ (corresponding to flat terms), the extension $\rho^{\star}$ extends this specification to arbitrary terms. Importantly, this extension interacts well with naturality conditions on $\rho$. This is the content of the next two lemmas:

\begin{lemma}\label{lem:rho-star-nat}
For every $f\colon X\to Y$, if the pre-GSOS law $\rho$ \eqref{eq:pre-gsos-law} is natural w.r.t.\ $\Sigmas f$, then its extension $\rho^{*}$ \eqref{eq:pre-gsos-law-ext} is natural w.r.t.\ $f$.
\end{lemma}

\begin{proof}
To show that the naturality condition $B\Sigmas f\comp \rho^{*}_X=\rho^{*}_Y\comp \Sigmas(f\times Bf)$ of $\rho^{*}$ w.r.t.\ $f$ holds, it suffices by the uniqueness statement of the recursion principle to prove that the morphisms on both sides make the diagram below commute:
\[
\begin{tikzcd}[column sep=20]
X\times BX \ar{rr}{\eta} \ar{d}[swap]{f\times Bf} & &  \Sigmas(X\times BX) \ar{d}{(-) } & \Sigma\Sigmas(X\times BX) \ar{l}[swap]{\ini}  \ar{d}{\Sigma\langle \Sigmas (f\comp \outl), \,(-)\rangle} \\
Y\times BY \ar{r}{\outr} & BY \ar{r}{B\eta} &  B\Sigmas Y & \Sigma(\Sigmas Y\times B\Sigmas Y)  \ar{l}[swap]{B\mu\comp \rho}
\end{tikzcd}
\]
For the morphism $\rho^{*}_Y\comp \Sigmas(f\times Bf)$, this follows from the commutative diagram below:
\[
\begin{tikzcd}[column sep=20]
X\times BX \ar{rr}{\eta} \ar{dd}[swap]{f\times Bf} & &  \Sigmas(X\times BX) \ar{d}{\Sigmas(f\times Bf)} & \Sigma\Sigmas(X\times BX) \ar{l}[swap]{\ini}  \ar{d}{\Sigma\Sigmas(f\times Bf)} \\
&& \Sigmas(Y\times BY) \ar{d}{\rho^{*}} & \Sigma\Sigmas(Y\times BY) \ar{d}{\Sigma\langle \Sigmas\outl,\,\rho^{*}\rangle} \ar{l}[swap]{\ini} \\
Y\times BY \ar{r}{\outr} \ar{urr}{\eta} & BY \ar{r}{B\eta} &  B\Sigmas Y & \Sigma(\Sigmas Y\times B\Sigmas Y)  \ar{l}[swap]{B\mu\comp \rho}
\end{tikzcd}
\]
The two upper cells commute by naturality of $\eta$ and $\ini$, respectively, and the two lower cells by definition of $\rho^{*}$.

Similarly, for the morphism $B\Sigmas f\comp \rho^{*}_X$, we get the commutative diagram below:
\[
\begin{tikzcd}[column sep=20]
X\times BX \ar{dr}{\outr} \ar{rr}{\eta} \ar{dd}[swap]{f\times Bf} & &  \Sigmas(X\times BX) \ar{d}{\rho^{*}} & \Sigma\Sigmas(X\times BX) \ar{l}[swap]{\ini}  \ar{d}{\Sigma\langle \Sigmas\outl,\,\rho^{*}\rangle} \\
& BX \ar{d}{Bf} \ar{r}{B\eta} & B\Sigmas X \ar{d}{B\Sigmas f} & \Sigma(\Sigmas X\times B\Sigmas X) \ar{d}{\Sigma(\Sigmas f\times B\Sigmas f)} \ar{l}[swap]{B\mu\comp \rho} \\
Y\times BY \ar{r}{\outr} & BY \ar{r}{B\eta} &  B\Sigmas Y & \Sigma(\Sigmas Y\times B\Sigmas Y)  \ar{l}[swap]{B\mu\comp \rho}
\end{tikzcd}
\]
The two upper cells commute by definition of $\rho_X^{*}$, and the three remaining cells by naturality of $\outr$, naturality of $\eta$, naturality of $\mu$, and naturality of $\rho$ w.r.t.\ $\Sigmas f$.
\end{proof}

\begin{lemma}\label{lem:gamma-rho-star}
If a pre-GSOS law $\rho$ is natural w.r.t.\ $\hat\ini$, then the following diagram commutes, where $\gamma$ is the operational model of $\rho$:
\[
\begin{tikzcd}[column sep=30]
 \mS \ar{d}[swap]{\gamma} && \Sigmas(\mS) \ar{ll}[swap]{\hat\ini}  \ar{d}{\Sigmas\langle \id,\,\gamma\rangle} \\
 B(\mS) & B\Sigmas(\mS) \ar{l}[swap]{B\hat\ini} & \Sigmas(\mS\times B(\mS)) \ar{l}[swap]{\rho^{*}}
\end{tikzcd}
\]

\end{lemma}

\begin{proof}
We show that the two morphisms $\gamma\comp \hat\ini$ and  $B\hat\ini\comp \rho^{*}_\mS\comp \Sigmas\langle \id,\gamma\rangle$ both make the diagram below commute; it then follows that they are equal by the uniqueness statement of the recursion principle.
\[
\begin{tikzcd}
\mS \ar{r}{\eta} \ar{dr}[swap]{\gamma} & \Sigmas(\mS) \ar{d}{(-)} && \Sigma\Sigmas(\mS) \ar{d}{\Sigma\langle \hat\ini, (-)\rangle} \ar{ll}[swap]{\ini} \\
& B(\mS) & B\Sigmas(\mS) \ar{l}[swap]{B\hat\ini} & \Sigma(\mS\times B(\mS)) \ar{l}[swap]{\rho}
\end{tikzcd}
\]
For $\gamma\comp \hat\ini$, we have the following diagram:
\[
\begin{tikzcd}[column sep=20]
\mS 
	\ar{r}{\eta} 
	\ar{dr}{\id} 
	\ar{ddr}[swap]{\gamma} & 
\Sigmas(\mS) 
	\ar{d}{\hat\ini} && 
\Sigma\Sigmas(\mS) 
	\ar{d}{\Sigma\hat\ini} 
	\ar{ll}[swap]{\ini} 
\\
& 
\mS 
	\ar{d}{\gamma} && 
\Sigma(\mS) 
	\ar{d}{\Sigma\langle \id,\,\gamma\rangle} 
	\ar{ll}[swap]{\ini} 
\\
& 
B(\mS) & 
B\Sigmas(\mS) 
	\ar{l}[swap]{B\hat\ini} & 
\Sigma(\mS\times B(\mS)) 
	\ar{l}[swap]{\rho}
\end{tikzcd}
\]
The two upper cells commute by definition of $\hat\ini$ as the free extension of $\id\colon \mS\to \mS$, and the lower cell by definition of $\gamma$.

For $B\hat\ini\comp \rho^{*}_\mS\comp \Sigmas\langle \id,\gamma\rangle$, we get the following commutative diagram.
\[
\begin{tikzcd}[column sep=1]
\mS 
	\ar{dr}{\langle \id,\gamma\rangle}   
	\ar{rr}{\eta} 
	\ar{ddd}[swap]{\gamma} &[-1em]&[1em] 
\Sigmas(\mS) 
	\ar{d}{\Sigmas\langle\id,\gamma\rangle} &&[1em] 
\Sigma\Sigmas(\mS) 
	\ar{d}{\Sigma\Sigmas\langle\id,\gamma\rangle} 
	\ar{ll}[swap]{\ini} 
\\
& 
\mS\times B(\mS)  
	\ar{r}{\eta} 
	\ar{ddl}[swap]{\outr} & 
\Sigmas(\mS\times B(\mS)) 
	\ar{d}{\rho^{*}} 
&& 
\Sigma\Sigmas(\mS\times B(\mS)) 
	\ar{d}{\Sigma\langle \Sigmas\outl,\,\rho^{*}\rangle} 
	\ar{ll}[swap]{\ini} 
\\
& & 
B\Sigmas(\mS) 
	\ar{d}{B\hat\ini} & 
B\Sigmas\Sigmas(\mS) 
	\ar{d}{B\Sigmas\hat\ini} 
	\ar{l}[swap]{B\mu} & 
\Sigma(\Sigmas(\mS)\times B\Sigmas(\mS)) 
	\ar{l}[swap]{\rho} 
	\ar{d}{\Sigma(\hat\ini \times B\hat \ini)} 
\\
B(\mS) 
	\ar[equals]{rr} 
	\ar{urr}{B\eta} && 
B(\mS) & 
B\Sigmas(\mS) 
	\ar{l}[swap]{B\hat\ini} & 
\Sigma(\mS\times B(\mS)) 
	\ar{l}[swap]{\rho}
\end{tikzcd}
\]
The lower right cell commutes by assumption, and all other cells either by naturality or by definition.
\end{proof}
We are ready to prove the congruence result for pre-GSOS laws:

\begin{proof}[Proof of \Cref{thm:cong}]
Every $B$-coalgebra $(X,c)$ induces a $B$-coalgebra structure $\ol{c}$ on~$\Sigmas X$ defined as follows:
\[ \ol{c} \,=\, (\,\Sigmas X \xto{\Sigmas\langle \id,c\rangle} \Sigmas(X \times BX) \xto{\rho^{*}_{X}} B\Sigmas X\,). \]
In particular, for the final coalgebra $(\nu B,{\tau})$ we obtain the $B$-coalgebra $(\Sigmas(\nu B),\ol{\tau})$. Thus, we get a $\Sigmas$-algebra structure on $\nu B$ (equivalent to a $\Sigma$-algebra structure) given by unique coalgebra morphism
\[ \alpha\colon (\Sigmas(\nu B), \ol{\tau})\to (\nu B, \tau). \] 
To prove that $\beh$ is a morphic congruence, it suffices to show that 
\[ \beh \colon (\mS,\ini)\to (\nu B,\alpha) \]
is a $\Sigma^{*}$-algebra morphism, that is, $\beh\comp \hat\ini = \alpha\comp \Sigma^{*}\beh$. The latter equality follows by finality of the coalgebra $\nu B$ once we show that all four morphisms in the square below are $B$-coalgebra morphisms:
\[
\begin{tikzcd}
(\Sigmas(\mS),\ol{\gamma}) \ar{d}[swap]{\Sigmas \beh} \ar{r}{\hat\ini} & (\mS,\gamma) \ar{d}{\beh} \\
(\Sigmas (\nu B), \ol{\tau}) \ar{r}{\alpha} & (\nu B,\tau)
\end{tikzcd}
\]
The morphisms $\beh$ and $\alpha$ are coalgebra morphisms by definition. That $\hat\ini$ is a coalgebra morphism is precisely the statement of \Cref{lem:gamma-rho-star}. Finally, $\Sigmas \beh$ is a coalgebra morphism by the commutative diagram below:
\[
\begin{tikzcd}[column sep=9ex, row sep=normal]
\Sigmas(\mS) 
	\ar[shiftarr={yshift=25}]{rr}{\ol{\gamma}} 
	\ar{r}{\Sigmas\langle \id,\gamma\rangle} 
	\ar{d}[swap]{\Sigmas \beh} & 
\Sigmas(\mS\times B(\mS)) 
	\ar{r}{\rho^{*}_{\mS}} 
	\ar{d}{\Sigmas(\beh\times B\beh)} & 
B\Sigmas(\mS) 
	\ar{d}{B\Sigmas \beh} 
\\
\Sigmas(\nu B) 
	\ar[shiftarr={yshift=-20}]{rr}[swap]{\ol{\tau}} 
	\ar{r}{\Sigmas\langle \id,\tau\rangle} & 
\Sigmas(\nu B\times B(\nu B)) 
	\ar{r}{\rho^{*}_{\nu B}} & 
B\Sigmas(\nu B)
\end{tikzcd}
\] 
Here the left hand square commutes because $\beh$ is a coalgebra morphism, and the right hand square because $\rho^{*}$ is natural w.r.t.\ $\beh$, which follows from the assumption that $\rho$ is natural w.r.t.\ $\Sigmas\beh$ and \Cref{lem:rho-star-nat}.
\end{proof}
 
\subsection[Details for Remark~\ref{rem:products}]{Details for \Cref{rem:products}}

\Cref{asm}\ref{asm1} entails that $\Kl(T)$ has binary products, which coincide with binary coproducts in $\Kl(T)$ and thus binary coproducts in $\C$~\cite[Thm.~19]{cj13}. Indeed, for $X,Y,Z\in \C$ we have
\begin{align*} \Kl(T)(Z,X+Y)&=\C(Z,T(X+Y)) \\ &\cong \C(Z,TX\times TY) \\ &\cong \C(Z,TX)\times \C(Z,TY) \\&= \Kl(T)(Z,X)\times \Kl(T)(Z,Y).  \end{align*}

\subsection[Proof of Proposition~\ref{lem:op-model-coincidence}]{Proof of \Cref{lem:op-model-coincidence}}
Clearly $\barrho$ is natural in $\C$ because all morphisms appearing in its definition are. To show that the operational models $\gamma$ of $\rho$ in $\Kl(T)$ and $\barrho$ in $\C$ coincide, let us instantiate the general definition of $\gamma$ in $\Kl(T)$:
\begin{equation*}
\begin{tikzcd}[column sep=25]
 \mS \ar[kleisliw]{d}[swap]{\gamma} && \barSigma(\mS) \ar[kleisliw]{ll}[swap]{J\ini}  \ar[kleisliw]{d}{\barSigma\langle \eta,\,\gamma\rangle} \\
 \barB(\mS) & \barB\barSigmas(\mS) \ar[kleisliw]{l}[swap]{\barB J\hat\ini} & \barSigma(\mS\times \barB(\mS)) \ar[kleisliw]{l}[swap]{\rho}
\end{tikzcd}
\end{equation*}
This expands as follows, using the concrete definitions of the involved structures:
\begin{equation*}
\begin{tikzcd}[column sep=normal, row sep=normal]
T\mS\ar[ddd,"\gamma"'] & & \mS\ar[ll, "\eta"'] & & \Sigma\mS\ar[ll,"\iota"']\dar["\Sigma\brks{\eta,\gamma}"]
\\
& & & & \Sigma(T\mS\times TB\mS)
\dar["\Sigma j"]
\\
& & & & \Sigma T(\mS+B\mS)
	\dar["\delta^{\Sigma}"]
\\
TB\mS
& & 
TB\Sigmas(\mS) 
	\ar[ll,"TB\hat\iota"']
& 
TTB\Sigmas(\mS) 
	\lar["\mu"']
& 
T\Sigma (\mS+B\mS)
\lar["T\rho"']
\end{tikzcd}
\end{equation*}
This diagram is exactly the one uniquely determining $\gamma$ as the operational 
model of the pre-GSOS law \eqref{eq:bar-rho} in $\C$. \qed

\subsection[Proof of Proposition~\ref{prop:affine-vs-natural}]{Proof of \Cref{prop:affine-vs-natural}}

By definition of the Kleisli extensions corresponding to the given distributive laws, commutativity of \eqref{eq:affine} is equivalent to commutativity of the following diagram in~$\Kl(T)$:
\begin{equation}\label{eq:affine-kleisli}
\begin{tikzcd}
 \barSigma \barB_0 P \ar[kleisliw]{r}{\barSigma \bar B_0 Jp} 
  & \barSigma\barB_0 TX \ar[kleisliw]{r}{\rho_{TX}} \ar[kleisliw]{d}[swap]{\barSigma\barB_0\id_{TX}} & \barB \barSigmas TX \ar[kleisliw]{d}{\barB\barSigmas\id_{TX}} \\
  & \barSigma\barB_0 \ar[kleisliw]{r}{\rho_X} X & \barB\barSigmas X 
\end{tikzcd}
\end{equation}
Here $J\colon \C\to \Kl(T)$ is the canonical left adjoint, and we regard the identity morphism $\id_{TX}\c TX\to TX$ of $\C$ as a morphism $\id_{TX}\colon TX\kleislito X$ of $\Kl(T)$.

Now suppose that \eqref{eq:affine} (equivalently \eqref{eq:affine-kleisli}) commutes. Consider the following diagram in $\Kl(T)$:
\begin{equation}\label{eq:natural-kleisli}
\begin{tikzcd}
\barSigma\barB_0 P 
	\ar[kleisliw]{r}{\rho_P} 
	\ar[kleisliw, shiftarr={xshift=-55}]{dd}[swap]{\barSigma\barB_0 p} 
	\ar[kleisliw]{d}[swap]{\barSigma\barB_0 Jp} & 
\barB\barSigmas P 
	\ar[kleisliw]{d}{\barB\barSigmas Jp} 
	\ar[kleisliw, shiftarr={xshift=55}]{dd}{\barB\barSigmas p}
\\
\barSigma\barB_0 T X 
	\ar[kleisliw]{r}{\rho_{TX}} 
	\ar[kleisliw]{d}[swap]{\barSigma\barB_0\id_{TX}} & 
\barB\barSigmas TX 
	\ar[kleisliw]{d}{\barB\barSigmas \id_{TX}} 
\\
\barSigma\barB_0 X 
	\ar[kleisliw]{r}{\rho_X}  & 
\barB\barSigmas X
\end{tikzcd}
\end{equation}
The left and right parts commute because
\[ p = (\begin{tikzcd} P \ar[kleisliw]{r}{Jp} & TX \ar[kleisliw]{r}{\id_{TX}} & X \end{tikzcd}) \qquad\text{in $\Kl(T)$} \]
by definition of Kleisli composition and one of the monad laws. The upper part commutes by naturality of $\rho$ in $\C$, and the lower part commutes when precomposed with $\barSigma\barB_0 Jp$ by~\eqref{eq:affine-kleisli}. It follows that the outside commutes, proving naturality of \eqref{eq:gsos-law-kleisli} w.r.t.\ $p$.

Conversely, if \eqref{eq:gsos-law-kleisli} is natural w.r.t.\ $p$, then the outside and the left, right, and upper part of \eqref{eq:natural-kleisli} commute. It follows that the lower part commutes when precomposed with $\barSigma\barB_0 Jp$. Thus \eqref{eq:affine-kleisli} (equivalently \eqref{eq:affine}) commutes.\qed

\subsection[Proof of Theorem~\ref{thm:cong-trace-positive}]{Proof of \Cref{thm:cong-trace-positive}}

We first establish a few auxiliary results. The first one follows by instantiating the refined congruence result for abstract GSOS (\Cref{thm:cong}) to the present setting:

\begin{proposition}\label{thm:cong-trace}
If the pre-De Simone law \eqref{eq:de-simone-law} is natural w.r.t.\
$\barSigmas \tr\colon \barSigmas(\mu\barSigma)\kleislito \barSigmas(\mu \barB)$, then
 $\tr\colon \mS\to T(\mu B)$ is a morphic $\Sigma$-congruence.
\end{proposition}

\begin{proof}
The pre-GSOS law \eqref{eq:gsos-law-kleisli} corresponding to the given pre-De Simone law \eqref{eq:de-simone-law} is natural w.r.t.\ $J\hatini\colon \barSigmas(\mu\barSigma)\to \mu\barSigma$ because this morphism is pure and \eqref{eq:de-simone-law} is natural in $\C$. It is natural w.r.t.~$\barSigmas \tr$ by assumption. The congruence theorem for abstract GSOS (\Cref{thm:cong}) thus shows that the morphism $\tr$ is a morphic congruence on the initial algebra $\mu\barSigma$, i.e.\ there exists a $\barSigma$-algebra structure $\alpha\colon \barSigma(\mu\barB)\kleislito \mu\barB$ such that $\tr\colon (\mu\barSigma,J\ini)\kleislito (\mu\barB,\alpha)$ is a $\barSigma$-algebra morphism. This states precisely that the diagram in $\C$ below commutes:
\[
\begin{tikzcd}
\Sigma(\mu\Sigma) \ar{rrr}{\ini} \ar{d}[swap]{\Sigma\tr} &&& \mu\Sigma \ar{d}{\tr} \\
\Sigma T(\mu B) \ar{r}{\delta^{\Sigma}} & T\Sigma(\mu B) \ar{r}{T\alpha} & TT(\mu B) \ar{r}{\mu} & T(\mu B)
\end{tikzcd}
\]
Thus $\tr$ carries a $\Sigma$-algebra morphism, proving that $\tr$ is a morphic congruence on~$\mu\Sigma$.
\end{proof}

Let us fix the affine part $T_{+}$ of $T$ and
the monad morphism $i\c T_{+}\to T$ constructed in \Cref{pro:affine-part}.

\begin{proposition}\label{pro:d-rest}
There exists a functor-over-monad distributive law $\delta^{\Sigma+}\colon \Sigma T_{+}\to T_{+}\Sigma$ 
that restricts the given law $\delta^{\Sigma}$, that is, the following square commutes:
\[
\begin{tikzcd}
\Sigma T_{+} \ar{r}{\delta^{\Sigma+}} \ar{d}[swap]{\Sigma i} & T_{+}\Sigma \ar{d}{i} \\
\Sigma T \ar{r}{\delta^{\Sigma}} & T\Sigma 
\end{tikzcd}
\] 
\end{proposition}
\begin{proof}
Consider the diagram
\begin{equation*}
\begin{tikzcd}[column sep=12ex, row sep=normal]
\Sigma T_{+}X
	\rar["\Sigma i_X"]
	\dar[dashed,"\delta^{\Sigma+}"'] &
\Sigma
	\dar["\delta^\Sigma"] TX
\\
T_{+}\Sigma X
	\rar["i_{\Sigma X}"] &
T\Sigma X
	\ar[r,shift right=.35ex,"T!"']
	\ar[r,shift left=.75ex,"\eta\comp !"] &
T 1
\end{tikzcd}
\end{equation*}
The dashed arrow is the component of the distributive law in question. To show that 
it indeed exist, by universality of equalizers, it suffices to prove that $(\eta\comp !)\comp\delta^\Sigma\comp\Sigma i_X = (T!)\comp\delta^\Sigma\comp\Sigma i_X$.
Indeed, $(\eta\comp !)\comp\delta^\Sigma\comp\Sigma i_X = \eta\circ \bang$, and 
\begin{align*}
T!\comp\delta^\Sigma\comp\Sigma i_X 
=&\; T!\comp (T\Sigma !)\comp\delta^\Sigma\comp\Sigma i_X\\*
=&\; T!\comp \delta^\Sigma\comp(\Sigma T!)\comp\Sigma i_X\\
=&\; T!\comp \delta^\Sigma\comp\Sigma(\eta\comp !)\comp\Sigma i_X\\
=&\; T!\comp\eta\comp \Sigma !\\
=&\; \eta\comp !\comp\Sigma !\\
=&\; \eta\comp !.
\end{align*}
The third step uses the unit axiom of the distributive law $\delta^\Sigma$. 
The axioms of a distributive law hold for $\delta^{\Sigma+}$ because they are inherited from those for
$\delta^\Sigma$ along the injection $i$.
\end{proof}

Recall that the free monad $\Sigma^{*}$ in $\C$ extends to the free monad $\barSigmas$ in $\Kl(T)$.

\begin{proposition}\label{prop:sigma-star-pres-positive}
$\barSigmas$ preserves affine morphisms.
\end{proposition}

\begin{proof}
The monad morphism $i\c T_{+}\to T$ induces a functor
\[ I\colon \Kl(T_{+}) \to \Kl(T) \]      
that acts as identity on objects, and is given on morphisms by
\[ (p\colon P\to T_{+}X)\quad\mapsto\quad (i_X\comp p\colon P\to TX).  \]
Note that $I$ commutes with the canonical left adjoints:
\begin{equation}\label{eq:i-commute}
\begin{tikzcd}
\Kl(T_{+}) \ar{rr}{I} && \Kl(T) \\
& \C \ar{ul}{J_{+}} \ar{ur}[swap]{J} &
\end{tikzcd}
\end{equation}
Consider the Kleisli extension
\[ \widetilde{\Sigma}\colon \Kl(T_{+})\to \Kl(T_{+})  \]
of $\Sigma$ induced by the restricted distributive law $\delta^{\Sigma+}$ from \Cref{pro:d-rest}. By \Cref{prop:initial-alg-lift}, we know that $\barSigma$ and $\widetilde{\Sigma}$ have free monads $\barSigmas$ and $\widetilde{\Sigma}^{*}$, both extending the free monad $\Sigmas$. From this and \eqref{eq:i-commute} it follows that
\[ I\widetilde{\Sigma}^{*}=\barSigmas I. \]  
Thus, for every affine morphism $p=Ip_{+}$ of $\Kl(T)$, we have
\[ \barSigmas P= \barSigmas Ip_{+} = I\widetilde{\Sigma}^{*}p_{+},  \]
which proves that $\barSigmas p$ is affine.
\end{proof}

We are prepared to prove \Cref{thm:cong-trace-positive}:

\begin{proof}[Proof of \Cref{thm:cong-trace-positive}]
Note first that a pre-De Simone law \eqref{eq:de-simone-law} is a De Simone law iff it is natural w.r.t.\ all affine morphisms. Since $\tr$ is affine by assumption, we know that $\barSigmas\tr$ is affine by \Cref{prop:sigma-star-pres-positive}. Therefore, by \Cref{thm:cong-trace}, we conclude that $\tr$ forms a morphic congruence.
\end{proof}

\subsection[Proof of Proposition~\ref{prop:tr-ineq}]{Proof of \Cref{prop:tr-ineq}}
Suppose that the morphism $j\comp\brks{\eta,\eta}\c 1\to T(1+1)$ and every component~\eqref{eq:de-simone-law} is affine. Let us first show that $T!\comp\gamma = \eta\comp!$, meaning affineness of $\gamma$ 
by \Cref{lem:positive-criterion}.
Consider the following diagram:
\begin{equation*}
\begin{tikzcd}[column sep=4em, row sep=normal]
\Sigma(\mS)
	\dar["\Sigma\brks{\id,\gamma}"']
	\ar[rr,"\iota"] 
	&[4em] &
\mS
	\dar["\brks{\id,\gamma}"]
\\
\Sigma(\mS\times TB\mS)
	\rar["\brks{\iota\comp\Sigma\fst,\bar\rho}"] 
	\dar["\Sigma(\id\times T!)"']
&
\mS\times TB\Sigmas\mS
	\rar["\id\times TB\mu"]
&
\mS\times TB\mS
	\dar["\id\times T!"]
\\
\Sigma(\mS\times T1)
	\rar["\brks{\iota\comp\Sigma\fst,\delta^\Sigma\comp \Sigma(T!\comp j\comp (\eta\times\id))}"] 
&
\mS\times T\Sigma 1
	\rar["\id\times T!"] &
\mS\times T1
\end{tikzcd}
\end{equation*}
The top cell commutes by \Cref{lem:op-model-coincidence}. The bottom cell corresponds to 
a pair of equations, one of which is trivial, and the other one amounts to
\begin{align*}
T!\comp TB\mu\comp\bar\rho = T!\comp\delta^\Sigma\comp\Sigma(T!\comp j\comp (\eta\times T!)).
\end{align*}
Using the construction~\eqref{eq:bar-rho} of $\bar\rho_X$:
\begin{align*}
T!\comp TB\mu\comp \bar\rho 
=&\; T!\comp\bar\rho\\  
=&\; T!\comp\mu\comp T\rho\comp\delta^\Sigma\comp\Sigma j\comp\Sigma(\eta\times\id)\\  
=&\; \mu\comp TT!\comp T\rho\comp\delta^\Sigma\comp\Sigma j\comp\Sigma(\eta\times\id)\\  
=&\; \mu\comp T\eta\comp T!\comp\delta^\Sigma\comp\Sigma j\comp\Sigma(\eta\times\id)\\  
=&\; T!\comp\delta^\Sigma\comp\Sigma j\comp\Sigma(\eta\times\id)\\  
=&\; T!\comp T\Sigma!\comp\delta^\Sigma\comp\Sigma j\comp\Sigma(\eta\times\id)\\  
=&\; T!\comp \delta^\Sigma\comp \Sigma T!\comp\Sigma j\comp\Sigma(\eta\times\id)\\  
=&\; T!\comp\delta^\Sigma\comp \Sigma T!\comp\Sigma j\comp\Sigma(\eta\times T!)
\end{align*}
using the assumption $T!\comp\rho_X = \eta\comp!$.
Thus, the entire diagram commutes. Since the diagram expresses the fact that 
$\brks{\id,T!\comp\gamma}$ is a unique $\Sigma$-algebra morphism, it suffices to show that the following diagram commutes to
conclude that $T!\comp\gamma = \eta\comp !$: 
\begin{equation*}
\begin{tikzcd}[column sep=3em, row sep=normal]
\Sigma\mS
	\dar["\Sigma\brks{\id,\eta\comp !}"']
	\ar[rr,"\iota"] 
	&[7em] &
\mS
	\dar["\brks{\id,\eta\comp !}"]
\\
\Sigma(\mS\times T1)
	\rar["\brks{\iota\comp\Sigma\fst,\delta^\Sigma\comp \Sigma(T!\comp j\comp (\eta\times\id))}"] 
&
\mS\times T\Sigma 1
	\rar["\id\times T!"] &
\mS\times T1
\end{tikzcd}
\end{equation*}
Again, the relevant equation to verify is
\begin{align*}
T!\comp \delta^\Sigma\comp \Sigma(T!\comp j\comp\brks{\eta,\eta}) = \eta\comp !,
\end{align*}
and indeed, 
\begin{align*}
T!\comp \delta^\Sigma\comp \Sigma(T!\comp j\comp\brks{\eta,\eta}) =&\;
T!\comp \delta^\Sigma\comp \Sigma(\eta\comp !)\\
=&\; T!\comp\eta\comp \Sigma !\\
=&\; \eta\comp !,
\end{align*}
where we used the assumption that $j\comp\brks{\eta,\eta}$ is affine, that is, $T!\comp j\comp\brks{\eta,\eta} = \eta\comp !$.

Let us proceed with showing the inequality $T!\comp\tr \sqsubseteq\eta\comp !$.
We use the fact that $\tr$ is uniquely determined by the diagram
\begin{equation*}
\begin{tikzcd}[column sep=3em, row sep=normal]
\mS
	\ar[rrr, "\tr"]
	\dar["\gamma"'] & & & T\mu B\dar["T\beta^{-1}"]
\\
TB\mS\rar["TB\tr"] & TBT\mu B\rar["T\delta^B"] & TTB\mu B\rar["\mu"] & TB\mu B                
\end{tikzcd}
\end{equation*}
This means that $\tr$ is the least fixpoint, or equivalently the least prefixpoint, of the monotone endomap on the DCPO $\Kl(T)(\mS,\mu B)$ given by 
\begin{align*}
f\quad\mapsto\quad T\beta\comp\mu\comp T(\delta^B\comp Bf)\comp\gamma.
\end{align*}
We use the following principle~\cite[Corollary 2.6]{amm21} to show that $T!\comp\tr\sqsubseteq\eta\comp !$: 

\begin{itemize}
  \item[]
\emph{Given a DCPO $P$ with $\bot$, if $\CF\c P\to P$ is monotone, then its least (pre)fixpoint $\mu\CF$
belongs to every subset $S\subseteq P$ which contains $\bot$ and is closed under $f$ 
and directed joins.}
\end{itemize}
We apply this result to 
\[P=\Kl(T)(\mS,\mu B),\qquad \CF(f) = T\beta\comp\mu\comp T(\delta^B\comp Bf)\comp\gamma,\qquad
S = \{\,f\mid T!\comp f\sqsubseteq \eta\comp !\,\}.\] Since $\C$ is $\DCPOb$-enriched and $\barB$ is locally monotone (\Cref{def:trace-setting}), we conclude that~$P$ is a DCPO
with $\bot$, and that $\CF$ is monotone. Obviously, $S$ is closed under directed joins. Moreover, $S$ contains $\bot$: we have $\bot\sqsubseteq \eta\comp !$ and therefore $T!\comp \bot \sqsubseteq T!\comp \eta \comp ! = \eta\comp !$. Let us show that $S$ is closed under $\CF$. Suppose that $f\in S$, i.e.\ $T!\comp f\sqsubseteq \eta\comp !$.
Then
\begin{align*}
T!\comp \CF(f) =& T!\comp T\beta\comp\mu\comp T(\delta^B\comp Bf)\comp\gamma \\
=&\; T!\comp\mu\comp T(\delta^B\comp Bf)\comp\gamma \\
=&\; T!\comp\mu\comp TTB!\comp T(\delta^B\comp Bf)\comp\gamma \\
=&\; T!\comp\mu\comp T(\delta^B\comp BT!\comp Bf)\comp\gamma \\
=&\; T!\comp \mu\comp T(\delta^B\comp B(T!\comp f))\comp\gamma\\
\sqsubseteq&\; T!\comp \mu\comp T(\delta^B\comp B(\eta\comp !))\comp\gamma\\
=&\; T!\comp \mu\comp T(\eta\comp B !)\comp\gamma\\
=&\; T!\comp TB !\comp\gamma\\
=&\; T!\comp \gamma\\
=&\; \eta\comp !
\end{align*}
Thus, $\CF(f)\in S$, and therefore $\tr\in S$, i.e.\ $T!\comp\tr\sqsubseteq\eta\comp !$.
\qed

\subsection[Proof of Proposition~\ref{prop:de-simone-rule-law-affine}]{Proof of \Cref{prop:de-simone-rule-law-affine}}

By \Cref{prop:affine-vs-natural}, our task is to shown for all $X\in \Set$, the two legs of the diagram
\[
\begin{tikzcd}
\Sigma(\Pow X+L\times \Pow X + 1) \ar{d}[swap]{\Sigma\delta^{B_0}} \ar{r}{\rho} & \Pow(L\times \Sigmas \Pow X +1) \ar{d}{\Pow B\delta^{\Sigmas}}  \\
\Sigma \Pow (X+L\times X + 1) \ar{d}[swap]{\delta^{\Sigma}} & \Pow (L\times \Pow \Sigmas X +1) \ar{d}{\Pow \delta^{B}} \\
\Pow \Sigma (X+L\times X + 1) \ar{r}{\Pow \rho} & \Pow \Pow (L\times \Sigmas X+1) \ar{d}{\mu} \\
& \Pow(L\times \Sigmas X + 1) 
\end{tikzcd}
\]
send every element $\f(U_1,\ldots,U_n)$ of $\Sigma(\Pow_{+} X+L\times \Pow_{+} X + 1)$ to the same element of $\Pow(L\times \Sigmas X+1)$. We consider two cases:

\medskip\noindent \underline{\emph{Case 1:}} $U_i=*$ for some $i$.

\smallskip\noindent
Then $\rho$ maps $\f(U_1,\ldots,U_n)$ to $\{*\}$, and $\delta^{B}\comp B\delta^{\Sigmas}$ maps $*$ to $\{*\}$, which means 
\[ \mu\comp \Pow\delta^{B}\comp \Pow B\delta^{\Sigmas}\comp \rho(\f(U_1,\ldots,U_n)) = \{*\}. \]
Similarly, $\delta^{\Sigma}\comp \Sigma\delta^{B_0}$ maps $\f(U_1,\ldots,U_n)$ to a non-empty subset of $\Sigma(X+L\times X+1)$ all whose elements are of the form $\f(-,\ldots,-)$ with $*$ in the $i$th component. Since $\rho$ maps any such element to $\{*\}$, we have
\[ \mu\comp P\rho\comp \delta^{\Sigma}\comp \Sigma\delta^{B_0}(\f(U_1,\ldots,U_n) = \{*\}. \] 
\underline{\emph{Case 2:}} $U_i\in \Pow X+L\times \Pow X$ for all $i$.

\smallskip\noindent For notational simplicity, let us assume that $U_i\in L\times \Pow_{+} X$ (say $U_i=(a_1,V_i)$) for $i=1,\ldots,k$, and $U_{j}\in \Pow_{+} X$ for $j=k+1,\ldots,n$. Then $\rho$ maps $\f(U_1,\ldots,U_n)$ to the set containing the element $*$ and all pairs \[(a,t[V_1/x_1,\ldots,V_k/x_k,U_{k+1}/x_{k+1},\ldots,U_n/x_n])\in L\times \Sigmas \Pow X\] such that
\begin{equation}\label{eq:de-simone-rule-example}
\frac{x_{1}\xto{a_1} y_{1} \quad\cdots\quad x_{k}\xto{a_k} y_{k}}
{\f(x_1,\ldots,x_n)\xto{a} t} 
\end{equation}
is a rule in $\R$. Each such pair is mapped by $\delta^{B}\comp B\delta^{\Sigmas}$ to the set of all pairs
$(a,t[v_1/y_1,\ldots,v_k/y_k,u_{k+1}/x_{k+1},\ldots,u_n/x_n])$ where $v_i\in V_i$ and $u_j\in U_j$. (Here affinity of $t$ is needed, see \Cref{rem:affine}.) Thus, overall,
\[  \mu\comp \Pow\delta^{B}\comp \Pow B\delta^{\Sigmas}\comp \rho(\f(U_1,\ldots,U_n)) \]     
is the subset of $L\times \Sigmas X + 1$ consisting of the following elements:
\begin{enumerate}
\item \label{pf1} the element $*$ of $1$;
\item\label{pf2} all $(a,t[v_1/y_1,\ldots,v_k/y_k,u_{k+1}/x_{k+1},\ldots,u_n/x_n])\in L\times \Sigmas X$ where \eqref{eq:de-simone-rule-example} is a rule in $\R$, $v_i\in V_i$, and $u_j\in U_j$.
\end{enumerate}

Similarly, $\f(U_1,\ldots,U_n)$ is mapped by $\delta^{\Sigma}\comp \Sigma\delta^{B_0}$ to the non-empty set of all $\f((a_1,v_1),\ldots,(a_k,v_k),u_{k+1},\ldots,u_n)$ where $v_i\in V_i$ and $u_j\in U_j$. Moreover $\rho$ maps each such pair to the set containing the element $*$ and all pairs $(a,t[v_1/y_1,\ldots,v_k/y_k,u_{k+1}/x_{k+1},\ldots,u_n/x_n])$ such that \eqref{eq:de-simone-rule-example} is a rule in $\R$. Thus the subset
\[ \mu\comp  \Pow\rho\comp \delta^{\Sigma}\comp \Sigma\delta^{B_0}(\f(U_1,\ldots,U_n)) \] 
of $L\times \Sigmas X+1$ also contains precisely the elements \ref{pf1} and \ref{pf2} described above.

\subsection[Details for Example~\ref{ex:prob-lts}]{Details for \Cref{ex:prob-lts}}

Consider a simple probabilistic process algebra with a terminating process~$\mathsf{nil}$, a prefixing operator $a.(-)$ for each $a\in L$, and a parallel composition operator $\parallel_r$ for each $r\in [0,1]$. The intention is that the process $p\parallel_r q$ progresses one of its subprocesses $p$ or $q$ with probability $r$ or $1-r$, respectively. Formally, the operational semantics are given by the a.s.t.\ probabilistic De Simone specification shown below:
  \begin{gather*}
    \inference{}{\mathsf{nil}\xto{1}*}
    \qquad
    \inference{}{a.p \xto{a/1} p }
    \qquad
    \inference{p\xto{a} p'}{p\parallel_r q \xto{a/r} p' \parallel_r q} \qquad 
    \inference{p\to *}{p\parallel_r q \xto{r} *} \\[1ex]
    \inference{q\xto{a} q'}{p\parallel_r q \xto{a/1-r} p\parallel_r q'} \qquad 
    \inference{q\to *}{p\parallel_r q \xto{1-r} *} 
  \end{gather*}
   That the operational model \eqref{eq:op-model-weighted} is indeed a probabilistic LTS is easily verified by structural induction. The only interesting case is that of parallel composition. Thus consider a term $p\parallel_r q$ and suppose that both
\[ \gamma(p) =  \sum_{a\in L}\sum_{p'\in \mS} r_{a,p'}\cdot p' + r_{*}\cdot * \qqand \gamma(q) = \sum_{a\in L}\sum_{q'\in \mS} s_{a,q'}\cdot q' + s_{*}\cdot *  \]
are probability distributions. By definition of \eqref{eq:op-model-weighted} we have that
\begin{align*} 
& \gamma(p\parallel_r q) \\
= &  \sum_{a\in L}\sum_{p'\in \mS} r\cdot r_{a,p'} \cdot (p'\parallel_r q) + r\cdot r_{*}\cdot * \\
&  + \sum_{a\in L}\sum_{q'\in \mS} (1-r)\cdot s_{a,q'} \cdot (p\parallel_r q') + (1-r)\cdot s_{*}\cdot * 
\end{align*}
The sum of all coefficients is 
\begin{align*} 
& r\cdot \big(\sum_{a\in L}\sum_{p'\in \mS} r_{a,p'} + r_{*}\big)\\
&  + (1-r)\cdot \big( \sum_{a\in L}\sum_{q'\in \mS} s_{a,q'}+ s_{*} \big)\\
=~& r\cdot 1 + (1-r)\cdot 1 \\
=~&1,
\end{align*}
which proves that $\gamma(p\parallel_r q)$ is a probability distribution. 

Moreover, the probabilistic LTS specified by the above rules is almost surely terminating. In fact, the LTS is well-founded, that is, there exists no infinite path $t=t_0\xto{a_1}t_1\xto{a_2}t_2 \cdots$ in the LTS all of whose transitions have positive probabilities. This is immediate from the rules by structural induction on $t$.

\subsection[Proof of Proposition~\ref{prop:de-simone-rule-law-affine-weighted}]{Proof of \Cref{prop:de-simone-rule-law-affine-weighted}}
Let $\R$ be a probabilistic De Simone specification. 
To prove that the corresponding pre-De Simone law \eqref{eq:de-simone-law-from-rules-weighted} is a De Simone law, by \Cref{prop:affine-vs-natural} we need to show that the two legs of the diagram below send any element $\f(U_1,\ldots,U_n)$ of $\Sigma(\D X+L\times \D X + 1)$ to the same element of $\W(L\times \Sigmas X+1)$. (Recall that $\D=T_{+}$.)
\[
\begin{tikzcd}
\Sigma(\W  X+L\times \W  X + 1) \ar{d}[swap]{\Sigma\delta^{B_0}} \ar{r}{\rho_{\W  X}} & \W (L\times \Sigmas \W  X +1) \ar{d}{\W  B\delta^{\Sigmas}}  \\
\Sigma \W  (X+L\times X + 1) \ar{d}[swap]{\delta^{\Sigma}} & \W  (L\times \W  \Sigmas X +1) \ar{d}{\W  \delta^{B}} \\
\W  \Sigma (X+L\times X + 1) \ar{r}{\W  \rho_X} & \W  \W  (L\times \Sigmas X+1) \ar{d}{\mu} \\
& \W (L\times \Sigmas X + 1) 
\end{tikzcd}
\]
By symmetry, we may assume that:
\begin{itemize}
\item $U_j\in \D X$ for $j=1,\ldots, k$, say $U_j=\sum_{i\in I_j} r_i^j\cdot u_i^j$. 
\item $U_j\in L\times \D X$ for $j=k+1,\ldots,m$, say $U_j=(a_j, \sum_{i\in I_j} r_i^j\cdot u_i^j)$.
\item $U_j=*$ for $j=m+1,\ldots,n$. 
\end{itemize}
Let $\R_1$ be the set of all rules in $\R$ of the form 
\begin{equation*}
\frac{x_{{k+1}}\xto{a_{k+1}} y_{k+1} \quad\cdots\quad x_{m}\xto{a_m} y_{m} \quad x_{m+1} \to * \quad\cdots\quad x_{n} \to *}
{\f(x_1,\ldots,x_n)\xto{a/r} t} 
\end{equation*}
Similarly, let $\R_2$ be the set of all rules in $\R$ of the form 
\begin{equation*}
\frac{x_{{k+1}}\xto{a_{k+1}} y_{k+1} \quad\cdots\quad x_{m}\xto{a_m} y_{m} \quad x_{m+1} \to * \quad\cdots\quad x_{n} \to *}
{\f(x_1,\ldots,x_n)\xto{r} *} 
\end{equation*}
We write $a_R$ and $t_R$ for the output label and output term of a rule $R\in \R_1$.

We first consider the lower leg $\mu\comp \W\rho_X\comp \delta^{\Sigma}\comp \Sigma\delta^{B_0}$ of the diagram:
\begin{enumerate}
\item The map $\delta^{\Sigma}\comp \Sigma\delta^{B_0}$ sends $\f(U_1,\ldots,U_n)$ to the element of $\W \Sigma (X+L\times X + 1)$ given by 
\[ \sum_{i_1\in I_1}\cdots\sum_{i_m\in I_m} r_{i_1}^1\cdot \cdots \cdot r_{i_m}^m\cdot \f(u_{i_1}^1,\ldots,u_{i_k}^k,(a_{k+1},u_{i_{k+1}}^{k+1}),\ldots, (a_{m},u_{i_{m}}^{m}), *,\ldots,*).  \]
\item The map $\mu\comp \W\rho_X$ sends this formal sum to the element of $\W (L\times \Sigmas X + 1)$ given by 
\scriptsize
\begin{align*} 
& \sum_{i_1\in I_1}\cdots\sum_{i_m\in I_m} r_{i_1}^1\cdot \cdots \cdot r_{i_m}^m\cdot \big(\sum_{R\in \R_1} r_{R}\cdot t_R[u_{i_1}^1/x_1,\ldots,u_{i_k}^k/x_k,u_{i_{k+1}}^{k+1}/y_{k+1},\ldots, u_{i_m}^m/y_m]\\
& +\sum_{R\in \R_2} r_R\cdot *\big)\\
=& \sum_{i_1\in I_1}\cdots\sum_{i_m\in I_m}\sum_{R\in \R_1} r_R \cdot r_{i_1}^1\cdot \cdots \cdot r_{i_m}^m\cdot t_R[u_{i_1}^1/x_1,\ldots,u_{i_k}^k/x_k,u_{i_{k+1}}^{k+1}/y_{k+1},\ldots, u_{i_m}^m/y_m]  \\
& + \sum_{i_1\in I_1}\cdots\sum_{i_m\in I_m}\sum_{R\in \R_2} r_R \cdot r_{i_1}^1\cdot \cdots \cdot r_{i_m}^m\cdot * 
\end{align*}
\normalsize
The part after `$+$' simplifies to
\scriptsize
\[ \sum_{i_1,\ldots, i_m\in I_m}\sum_{R\in \R_2} r_R \cdot r_{i_1}^1\cdot \cdots \cdot r_{i_m}^m\cdot * = \sum_{R\in \R_2}r_R\cdot (\sum_{i_1,\ldots,i_m} r_{i_1}^1\dots r_{i_m}^m)\cdot * = \sum_{R\in \R_2} r_R\cdot * \]  
\normalsize
because the given formal sums in $U_1,\ldots,U_m$ are probability distributions and $r_{i_1}^1\dots r_{i_m}^m$ correspond to the probabilities arising in their product distribution. Therefore, overall, the lower leg of the diagram maps $\f(U_1,\ldots,U_n)$ to 
\scriptsize
\begin{equation}\label{eq:proof-term-1}
  \sum_{i_1,\ldots,i_m}\sum_{R\in \R_1} r_R \cdot r_{i_1}^1\cdot \cdots \cdot r_{i_m}^m\cdot t_R[u_{i_1}^1/x_1,\ldots,u_{i_k}^k/x_k,u_{i_{k+1}}^{k+1}/y_{k+1},\ldots, u_{i_m}^m/y_m] + \sum_{R\in \R_2} r_R\cdot *.
\end{equation}
\normalsize
\end{enumerate}
Similarly, for the upper leg $\mu\comp \W\delta^{B}\comp \W B\delta^{\Sigmas}\comp \rho_{\W X}$ of the diagram, we observe:
\begin{enumerate}
\item The map $\rho_{\W X}$ sends $\f(U_1,\ldots,U_n)$ to the element of $\W(L\times \Sigmas\W X+1)$ given by 
\[ \sum_{R\in \R_1} r_{R}\cdot t_R[U_1/x_1,\ldots,U_k/x_k,U_{k+1}/y_{k+1},\ldots,U_m/y_m] + \sum_{R\in \R_2} r_R\cdot *. \]  
\item Since all the terms $t_R$ are affine, the map $\mu\comp \W\delta^{B}\comp \W B\delta^{\Sigmas}$ sends this formal sum to the element of $\W (L\times \Sigmas X + 1)$ given by
\scriptsize
\[ \sum_{R\in \R_1} r_{R}\cdot \sum_{i_1,\ldots,i_m} r_{i_1}^1\cdot \cdots \cdot r_{i_m}^m \cdot t_R[u_{i_1}^1/x_1,\ldots,u_{i_k}^k/x_k,u_{i_{k+1}}^{k+1}/y_{k+1},\ldots,u_{i_m}^m/y_m] + \sum_{R\in \R_2} r_R\cdot *,
\]
\normalsize
which is equal to \eqref{eq:proof-term-1}.\qed
\end{enumerate}

\end{document}